\documentclass[
final
]{dmtcs-episciences}

\received{2015-12-22}
\accepted{2017-5-5}
\revised{2016-9-15}
\publicationdetails{19}{2017}{1}{13}{1350}
\usepackage{amsmath}
\usepackage{amssymb}
\usepackage{color}
\definecolor{red}{RGB}{255,0,0}
\definecolor{blue}{RGB}{0,0,255}
\definecolor{green}{RGB}{0,255,0}

\newcommand {\abs}[1]  {\left\vert#1\right\vert}
\newcommand {\set}[1]  {\left\{#1\right\}}

\newcommand {\defined} {\stackrel{def} {=}}

\newcommand {\npc}     {\textsc{NP}\textrm{-complete}}
\newcommand {\nph}     {\textsc{NP}\textrm{-hard}}

\newcommand {\bigoh}   {{\mathcal O}}

\newcommand{\commentfig}[1] {#1}

\newcommand {\itemref}[1] {\ref{itm:#1}\emph{)}}

\newtheorem{theorem}      {Theorem}     [section]
\newtheorem{corollary}    {Corollary}   [section]
\newtheorem{lemma}        {Lemma}       [section]
\newtheorem{proposition}  {Proposition} [section]
\newtheorem{obs}          {Observation} [section]

\usepackage{algorithm}
\usepackage{algorithmicx}
\usepackage[noend]{algpseudocode}
\usepackage{algpascal}
\usepackage[utf8]{inputenc}
\usepackage[round]{natbib}

\newcommand{\epu} {\textsc{EP}}
\newcommand{\ept} {\textsc{EPT}}
\newcommand{\epg} {\textsc{EPG}}
\newcommand{\enpu} {\textsc{ENP}}
\newcommand{\enpt} {\textsc{ENPT}}
\newcommand{\enpg} {\textsc{ENPG}}

\newcommand{\epug}[1] {\textsc{Ep}(#1)}
\newcommand{\eptg}[1] {\textsc{Ept}(#1)}
\newcommand{\epgg}[1] {\textsc{Epg}(#1)}
\newcommand{\enpug}[1] {\textsc{Enp}(#1)}
\newcommand{\enptg}[1] {\textsc{Enpt}(#1)}
\newcommand{\enpgg}[1] {\textsc{Enpg}(#1)}

\newcommand{\epugp} {\epug{\pp}}
\newcommand{\eptgp} {\eptg{\pp}}
\newcommand{\epggp} {\epgg{\pp}}
\newcommand{\enpugp} {\enpug{\pp}}
\newcommand{\enptgp} {\enptg{\pp}}
\newcommand{\enpggp} {\enpgg{\pp}}

\newcommand{\bepg}[1] {\textsc{B}_{{#1}}\textsc{-EPG}}
\newcommand{\benpg}[1] {\textsc{B}_{{#1}}\textsc{-ENPG}}
\newcommand{\boneepg} {\bepg{1}}
\newcommand{\boneenpg} {\benpg{1}}

\newcommand{\btwoenpg} {\benpg{2}}

\newcommand{\rept} {\left< T,\pp \right>}
\newcommand{\rep} {\left< H,\pp \right>}
\newcommand{\repprime} {\left< H',\pp' \right>}

\renewcommand{\split} {\textit{split}}

\newcommand{\pp} {{\cal P}}

\newcommand{\cs} {{\cal S}}

\author{Arman Boyac{\i} \affiliationmark{1}
  \and T{\i}naz Ekim \affiliationmark{1}
  \and Mordechai Shalom \affiliationmark{1,2} \thanks{Part of this work is accomplished while this author was visiting Bogazici University, Department of Industrial  Engineering, under the TUBITAK 2221 Program whose support is greatly acknowledged.}
  \and Shmuel Zaks \affiliationmark{3}}
\title{Graphs of Edge-Intersecting and Non-Splitting One Bend Paths in a Grid \thanks{This work was supported in part by TUBITAK PIA BOSPHORUS Grant No. 111M303.}
}
\affiliation{
  Department of Industrial Engineering, Bogazici University, Istanbul, Turkey \\
  TelHai College, Upper Galilee, 12210, Israel\\
  Department of Computer Science, Technion, Haifa, Israel}
\begin{document}
\maketitle
\begin{abstract}
The families $\ept$ (resp. $\epg$) Edge Intersection Graphs of Paths in a tree (resp. in a grid) are well studied graph classes. Recently we introduced the graph classes Edge-Intersecting and Non-Splitting Paths in a Tree ($\enpt$), and in a Grid ($\enpg$). It was shown that $\enpg$ contains an infinite hierarchy of subclasses that are obtained by restricting the number of bends in the paths. Motivated by this result, in this work we focus on one bend $\enpg$ graphs. We show that one bend $\enpg$ graphs are properly included in two bend $\enpg$ graphs. We also show that trees and cycles are one bend $\enpg$ graphs, and characterize the split graphs and co-bipartite graphs that are one bend $\enpg$. We prove that the recognition problem of one bend $\enpg$ split graphs is $\npc$ even in a very restricted subfamily of split graphs. Last we provide a linear time recognition algorithm for one bend $\enpg$ co-bipartite graphs.

\end{abstract}

\section{Introduction}\label{sec:intro}
\subsection{Background}
Given a host graph $H$ and a set $\pp$ of paths in $H$, the Edge Intersection Graph of Paths ($\epu$ graph)
of $\pp$ is denoted by $\epugp$. The graph $\epugp$ has a vertex for each path in $\pp$, and two vertices of $\epugp$ are adjacent if the corresponding two paths intersect in at least one edge. A graph $G$ is $\epu$ if there exist a graph $H$ and a set $\pp$ of  paths in $H$ such that $G=\epugp$. In this case, we say that $\rep$ is an $\epu$ representation of $G$. We also denote by $\epu$ the family of all graphs $G$ that are $\epu$.

The main application area of $\epu$ graphs is communication networks. Messages to be delivered are sent through routes of a communication network. Whenever two paths use the same link on the communication network, we say that they conflict. Noting that this conflict model is equivalent to an $\epu$ graph, several optimization problems in communication networks (such as message scheduling) can be seen as graph problems (such as vertex coloring) in the corresponding  $\epu$ graph.

In many applications it turns out that the host graphs are restricted to certain families such as paths, cycles, trees, grids, etc. Several known graph classes are obtained with such restrictions: when the host graph is restricted to paths, cycles, trees and grids, we obtain interval graphs, circular-arc graphs,  Edge Intersection Graph of Paths in a Tree ($\ept$) (see \cite{GJ85}), and Edge Intersection Graph of Paths in a Grid ($\epg$) (see \cite{DBLP:journals/networks/GolumbicLS09}), respectively.

Given a representation $\rept$ where $T$ is a tree and $\pp$ is a set of paths of $T$, the graph of edge intersecting and non-splitting paths of  $\rept$ (denoted by $\enptgp$) is defined as follows in \cite{BESZ13-ENPT1-DAM}:
$\enptgp$ has a vertex $v$ for each path $P_v$ of $\pp$ and two vertices $u,v$ of this graph are adjacent if the paths $P_u$ and $P_v$ edge-intersect and do not split (that is, their union is a path). We note that  $\enptgp$ is a subgraph of $\eptgp$. The motivation to study these graphs arises from all-optical Wavelength Division Multiplexing (WDM) networks in which two streams of signals can be transmitted using the same wavelength only if the paths corresponding to these streams do not split from each other (see \cite{BESZ13-ENPT1-DAM} for a more detailed discussion). A graph $G$ is an $\enpt$ graph if there is a tree $T$ and a set of paths $\pp$ of $T$ such that $G=\enptgp$. Clearly, when $T$ is a path, $\eptgp=\enptgp$ and this graph is an interval graph. Therefore, interval graphs are included in the class $\enpt$. In \cite{BESZ14-ENPG-TCS} we obtain the so-called $\enpu$ graphs by extending this definition to the case where the host graph is not necessarily a tree. In the same work, it has been shown that $\enpu=\enpg$ where $\enpg$ is the family of $\enpu$ graphs where the host graphs are restricted to grids. Whenever the host graph is a grid, it is common to use the following notion: a \emph{bend} of a path on a grid is an internal point in which the path changes direction. An $\enpg$ graph is $\benpg{k}$ if it has a representation in which every path has at most $k$ bends.

\subsection{Related Work}
While $\enpt$ and $\enpg$ graphs have been recently introduced,  $\ept$ and $\epg$ graphs are well studied in the literature. The recognition of $\ept$ graphs is NP-complete (\cite{Golumbic1985151}), whereas one can solve in polynomial time the maximum clique (\cite{Golumbic1985151}) and the maximum stable set (\cite{RobertE1985221}) problems in this class.

Several recent papers consider the edge intersection graphs of paths on a grid. Since all graphs are $\epg$ (see \cite{DBLP:journals/networks/GolumbicLS09}), most of the studies focus on the sub-classes of $\epg$ obtained by limiting the number of bends in each path.
An $\epg$ graph is $\bepg{k}$ if it admits a representation in which every path has at most $k$ bends.
The work of \cite{DBLP:journals/dmtcs/BiedlS10} investigates the minimum number $k$ such that $G$ has a $\bepg{k}$ representation for some special graph classes.
The work of \cite{DBLP:journals/networks/GolumbicLS09} studies the $\boneepg$ graphs. In particular it is shown that every tree is $\boneepg$, and a characterization of $C_4$ representations is given.
In \cite{DBLP:journals/dmtcs/BiedlS10} the existence of an outer-planar graph which is not $\boneepg$ is shown.
The recognition problem of $\boneepg$ graphs is shown to be $\npc$ in \cite{HKU10}. Similarly, in the class of $\boneepg$, the minimum coloring and the maximum stable set problems are $\npc$ (\cite{EGM2013}), however one can solve in polynomial time the maximum clique problem (\cite{EGM2013}).
\cite{AR2012} give a characterization of graphs that are both $\boneepg$ and belong to some subclasses of chordal graphs. Recently, \cite{CCH16} consider subclasses of $\boneepg$ obtained by restricting the representations to contain only certain subsets of the four possible single bend rectilinear paths. It is shown that for each possible non-empty subset of these four shapes, the recognition of the corresponding subclass of $\boneepg$ is an $\npc$ problem.

In \cite{BESZ13-ENPT1-DAM} we defined the family of $\enpt$ graphs and investigated the representations of induced cycles. These representations turn out to be much more complex than their counterpart in the $\ept$ graphs (discussed in \cite{GJ85}). In \cite{BESZ14-ENPG-TCS} we extended this definition to the general case in which the host graph is not necessarily a tree. We showed that the family of $\enpu$ graphs coincides with the family of $\enpg$ graphs, and that unlike $\epg$ graphs, not every graph is $\enpg$. We also showed that, in a way similar to the family of $\epg$ graphs, the sub families $\benpg{k}$ of $\enpg$ contains an infinite subset totally ordered by proper inclusion.

\subsection{Our Contribution}
In this work, we consider $\boneenpg$ graphs. In Section~\ref{sec:prelim} we present definitions and preliminary results among which we show that cycles and trees are $\boneenpg$ graphs. In Section~\ref{sec:split} we show that the $\boneenpg$ recognition problem is $\npc$ even for a very restricted subfamily of split graphs, i.e. graphs whose vertex sets can be partitioned into a clique and an independent set. In Section~\ref{sec:cobipartite} we show that $\boneenpg$ graphs can be recognized in  polynomial time within the family of co-bipartite graphs. As a byproduct, we also show that, unlike $\bepg{k}$ graphs, $\benpg{k}$ graphs do not necessarily admit a representation where every path has exactly $k$ bends. We summarize and point to further research directions in Section~\ref{sec:summary}.

\section{Preliminaries}\label{sec:prelim}
Given a simple graph (no loops or parallel edges) $G=(V(G),E(G))$ and a vertex $v$ of $G$, we denote by $N_G(v)$ the set of neighbors of $v$ in $G$, and by $d_G(v)=\abs{N_G(v)}$ the degree of $v$ in $G$. A graph is called \emph{$d$-regular} if every vertex $v$ has $d(v)=d$. Whenever there is no ambiguity we omit the subscript $G$ and write $d(v)$ and $N(v)$. Given a graph $G$ and $U \subseteq V(G)$, $N_U(v) \defined N_G(v) \cap U$. Two adjacent (resp. non-adjacent) vertices $u,v$ of $G$ are \emph{twins} (resp. \emph{false twins}) if $N_G(u) \setminus \set{v}=N_G(v) \setminus \set{u}$. For a graph $G$ and $U \subseteq V(G)$, we denote by $G[U]$ the subgraph of $G$ induced by $U$.

A vertex set $U \subseteq V(G)$ is a clique (resp. stable set) (of $G$) if every pair of vertices in $U$ is adjacent (resp. non-adjacent). A graph $G$ is a \emph{split graph} if $V(G)$ can be partitioned into a clique and a stable set. A graph $G$ is \emph{co-bipartite} if $V(G)$ can be partitioned into two cliques. Note that these partitions are not necessarily unique. We denote bipartite, co-bipartite and split graphs as $X(V_1,V_2, E)$ where
\begin{enumerate}[a)]
\item $X = B$ (resp. $C,S$) whenever $G$ is bipartite (resp. co-bipartite, split),
\item $V_1 \cap V_2 = \emptyset$, $V_1 \cup V_2 = V(G)$,
\item for bipartite graphs $V_1,V_2$ are stable sets,
\item for co-bipartite graphs $V_1$ and $V_2$ are cliques,
\item for split graphs $V_1$ is a clique and $V_2$ is a stable set, and
\item $E \subseteq V_1 \times V_2$ (in other words $E$ does not contain the cliques' edges).
\end{enumerate}
Unless otherwise stated we assume that $G$ is connected and both $V_1$ and $V_2$ are non-empty.

In this work every single path is simple, i.e. without duplicate vertices. However, if a union of paths is a path, the resulting path is not necessarily simple. For example, consider a graph on $5$ vertices $v_1,v_2,v_3,v_4,v_5$ and 5 edges $e_1=v_1 v_2, e_2=v_2 v_3, e_3=v_3 v_4, e_4=v_4 v_2, e_5=v_2 v_5$. Each of the paths $P_1 = \set{e_1 e_2 e_3}$ and $P_2 = \set{e_3, e_4, e_5}$ is simple. On the other hand, though $P_1 \cup P_2$ is a path, it is not simple. Whenever $v$ is an internal vertex of a path $P$, we sometimes say that $P$ \emph{crosses} $v$. Given two paths $P,P'$, a \emph{split} of $P,P'$ is a vertex with degree at least $3$ in $P \cup P'$. We denote by $\split(P,P')$ the set of all splits of $P$ and $P'$. When $\split(P, P') \neq \emptyset$ we say that $P$ and $P'$ are \emph{splitting}. Whenever $P$ and $P'$ edge intersect and $\split(P, P') = \emptyset$  we say that $P$ and $P'$ are \emph{non-splitting} and denote this by $P \sim P'$.  Clearly, for any two paths $P$ and $P'$ exactly one of the following holds: 
\begin{enumerate}[i)]
\item $P$ and $P'$ are edge disjoint,
\item $P$ and $P'$ are splitting, 
\item $P \sim P'$.
\end{enumerate}

A two-dimensional \emph{grid graph}, also known as a square grid graph, is an $m\times n$ lattice graph that is the Cartesian product graph of two paths $P$ and $P'$ of respectively length $n$ and $m$. Such a grid has vertex set $V=[n]\times [m]$.
A \emph{bend} of a path $P$ in a grid $H$ is an internal vertex of $P$ whose incident edges (in the path) have different directions, i.e. one vertical and one horizontal.

Let $\pp$ be a set of paths in a graph $H$. The graphs $\epugp$ and $\enpugp$ are such that $V(\enpugp) = V(\epugp) = V$, and there is a one-to-one correspondence between $\pp$ and $V$, i.e. $\pp=\set{P_v: v \in V}$. Given two paths $P_u,P_v \in \pp$, $\set{u,v}$ is an edge of $\epugp$ if and only if $P_u$ and $P_v$ have a common edge (cases (ii) and (iii)), whereas $\set{u,v}$ is an edge of $\enpugp$ if and only if $P_u \sim P_v$ (case (iii)). Clearly, $E(\enpugp) \subseteq E(\epugp)$. A graph $G$ is $\enpu$ if there is a graph $H$ and a set of paths $\pp$ of $H$ such that $G=\enpugp$. In this case $\rep$ is an $\enpu$ \emph{representation} of $G$. When $H$ is a tree (resp. grid) $\epugp$ is an $\ept$ (resp. $\epg$) graph, and $\enpugp$ is an $\enpt$ (resp. $\enpg$) graph; these graphs are denoted also as $\eptgp$, $\epggp$, $\enptgp$ and $\enpggp$, respectively. We say that two representations are \emph{equivalent} if they are representations of the same graph.

Let $\rep$ be a representation of an $\enpu$ graph $G$. For each edge $e$ of $H$, $\pp_e$ denotes the set of the paths of $\pp$ containing the edge $e$, i.e. $\pp_e \defined \set{P \in \pp|~e \in P}$. For a subset $U \subseteq V(G)$ we define $\pp_U \defined \set{P_v \in \pp : v \in U}$. Following standard notations, $\cup \pp_U \defined \cup_{P \in \pp_U} P$.

Given two paths $P$ and $P'$ of a graph, a \emph{segment} of $P \cap P'$ is a maximal path that constitutes a sub-path of both $P$ and $P'$. Clearly, $P \cap P'$ is the union of edge disjoint segments. We denote the set of these segments by $\cs(P,P')$.

Throughout the paper, whenever a representation $\rep$ of an $\enpg$ graph is given, we assume the host graph $H$ is a grid on sufficiently many vertices each of which is denoted by an ordered pair of integers.

The following Proposition that is proven in \cite{BESZ14-ENPG-TCS} is the starting point of many of our results.
\begin{proposition}\label{prop:BendsInClique}
\cite{BESZ14-ENPG-TCS} Let $K$ be a clique of a $\boneenpg$ graph $G$ with a representation $\rep$. Then $\cup \pp_K$ is a path with at most $2$ bends. Moreover, there is an edge $e_K \in E(H)$ such that every path of $\pp_K$ contains $e_K$.
\end{proposition}
Note that whenever $\cup \pp_K$ has two bends, $e_K$ lies between these two bends. Based on the above proposition, given two cliques $K,K'$ of a $\boneenpg$ graph we denote $\cs(K,K') \defined \cs(\cup \pp_K, \cup \pp_{K'})$.

By the following two observations, in the sequel we focus on connected twin-free graphs.
\begin{obs}\label{obs:twin}
Let $G$ be a graph and $G'$ obtained from $G$ by removing a twin vertex until no twins remain. Then, $G$ is $\benpg{k}$ if and only if $G'$ is $\benpg{k}$.
\end{obs}

\begin{obs}\label{obs:connected}
A graph $G$ is $\benpg{k}$ if and only if every connected component of $G$ is $\benpg{k}$.
\end{obs}

We first observe that some well-known graph classes are included in $\boneenpg$.

\begin{proposition}\label{prop:SomeOneBendENPGGraphs}
~\\
i) Every cycle is $\boneenpg$.\\
ii) Every tree is $\boneenpg$, and it has a representation $\rep$ where $\cup \pp$ is a tree.
\end{proposition}

\begin{proof}
\begin{enumerate}[i)]

\item For $k=3$ three identical paths consisting of one edge constitutes a $\boneenpg$ representation of $C_3$. For $k=4$ Figure~\ref{fig:B1ENPG-cycle} (a) depicts a $\boneenpg$ representation of $C_4$. Finally for any $k > 4$, we can construct a $C_k$ as shown in Figure~\ref{fig:B1ENPG-cycle} (b) for the case $k=11$.

\begin{figure}
\centering
\commentfig{\includegraphics{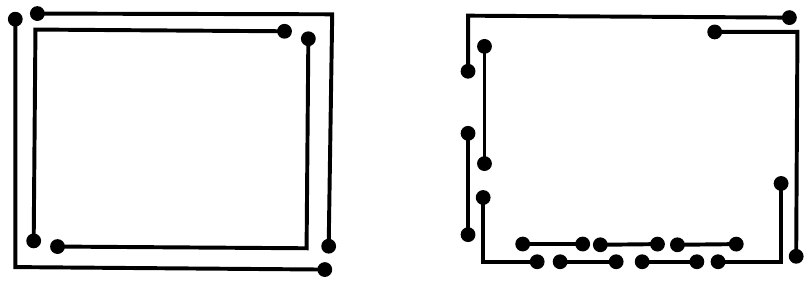}}
\caption{(a) A $\bepg{1}$ representation of $C_4$, (b) A $\bepg{1}$ representation of $C_{11}$.}
\label{fig:B1ENPG-cycle}
\end{figure}

\item Given a representation $\rep$ of a $\boneenpg$ graph $G$ and $U \subseteq V(G)$, we denote by $R_U$ the bounding rectangle of $\pp_U$. Let $T$ be a tree with a root $r$. We prove the following claim by induction on the structure of $T$ (see Figure~\ref{fig:B1ENPG-tree}). The tree $T$ has a $\boneenpg$ representation $\rep$ in which the corners of the bounding rectangle $R_T$ can be renamed as $a_T,b_T,c_T,d_T$ in counterclockwise order such that i) every path of $\pp$ has exactly one bend, ii) $b_T$ is a bend of $P_r$, iii) $a_T$ is an endpoint of $P_r$, iv) the line between $a_T$ and $d_T$ is used exclusively by $P_r$, v) $\cup \pp$ is a tree.

If $T$ is an isolated vertex, any path with one bend is a representation of $T$. Moreover, it is easy to verify that it satisfies conditions i) through v).

Otherwise let $T_1,\ldots,T_k$ be the subtrees of $T$ obtained by the removal of $r$, with roots $r_1,\ldots,r_k$ respectively. By the inductive hypothesis every such subtree $T_i$ has a representation with bounding box $a_{T_i},b_{T_i},c_{T_i},d_{T_i}$ satisfying conditions i) through iv). We now build a representation of $T$ satisfying the same conditions. We shift and rotate the representations of $T_1,\ldots,T_k$ so that the bounding rectangles do not intersect and the vertices $a_{T_1},b_{T_1},a_{T_2},b_{T_2},\ldots,a_{T_k},b_{T_k}$ are on the same horizontal line and in this order (See Figure~\ref{fig:B1ENPG-tree}). We extend the paths $P_{r_2},\ldots,P_{r_k}$ representing the roots of the trees $T_2,\ldots,T_k$ such that the endpoint $a_{T_i}$ of $P_{r_i}$ is moved to $a_{T_1}$.

Since $a_{T_i}$ is used exclusively by $P_{r_i}$ this modification does not cause $P_{r_i}$ to split from a path of $\pp_{V(T_i)}$. Therefore, the individual trees $T_1,\ldots,T_K$ are properly represented. Clearly, if two paths from different subtrees $T_i,T_j$ ($i<j$) intersect, then one of the intersecting paths must be $P_{r_j}$. The path $P_{r_j}$ intersects the bounding rectangle of $T_i$ only at the path between $a_i$ and $b_i$. As every path of $\pp_{V(T_i)}$, in particular one intersecting $P_{r_j}$ has one bend, such a path splits from $P_{r_j}$. Therefore, for any pair of vertices $(v_i,v_j) \in T_i \times T_j$ we have that $v_i$ and $v_j$ are non-adjacent in $\enpg(\pp)$, as required.

We rename the corners of the bounding rectangle $R_T$ such that $b_T = a_{T_1}$. We now add the path $P_r$ from $b_{T_1}$ to $a_T$ with a bend at $b_T$. The conditions i), ii), iii) are satisfied. We extend $P_r$ by one edge at $a_T$ to make sure that the line between $a_T$ and $d_T$ is exclusively used by $P_r$, thus satisfying condition iv). The extension of the paths $P_{r_2},\ldots,P_{r_k}$ does not add new edges to regions bounded by $a_{T_i},b_{T_i},c_{T_i},d_{T_i}$ and they don't introduce cycles between this regions. Moreover, since the line between $a_{T_1}$ and $d_{T_1}$ is used only at $a_{T_1}$ the path $P_r$ doe s not introduce any cycles either. Therefore, $\cup \pp$ is a tree, i.e. condition v) is satisfied.

The path $P_r$ intersects only $R_{T_1}$. This intersection is the path between $b_{T_1}$ and $d_{T_1}$ bending at $a_{T_1}$. Every path that intersects $P_r$ and does not split from it must bend at $a_{T_1}$. As $a_{T_1}$ is used exclusively by $P_{r_1}$, $P_{r_1}$ is the only path that possibly satisfies $P_{r_1} \sim P_r$. We now observe that $P_{r_i} \sim P_r$ for every $i \in [k]$. Therefore $r$ is adjacent to the root of $T_j$ in $\enpg(\pp)$, as required. \begin{figure}
\centering
\commentfig{\includegraphics[width=1.0\textwidth]{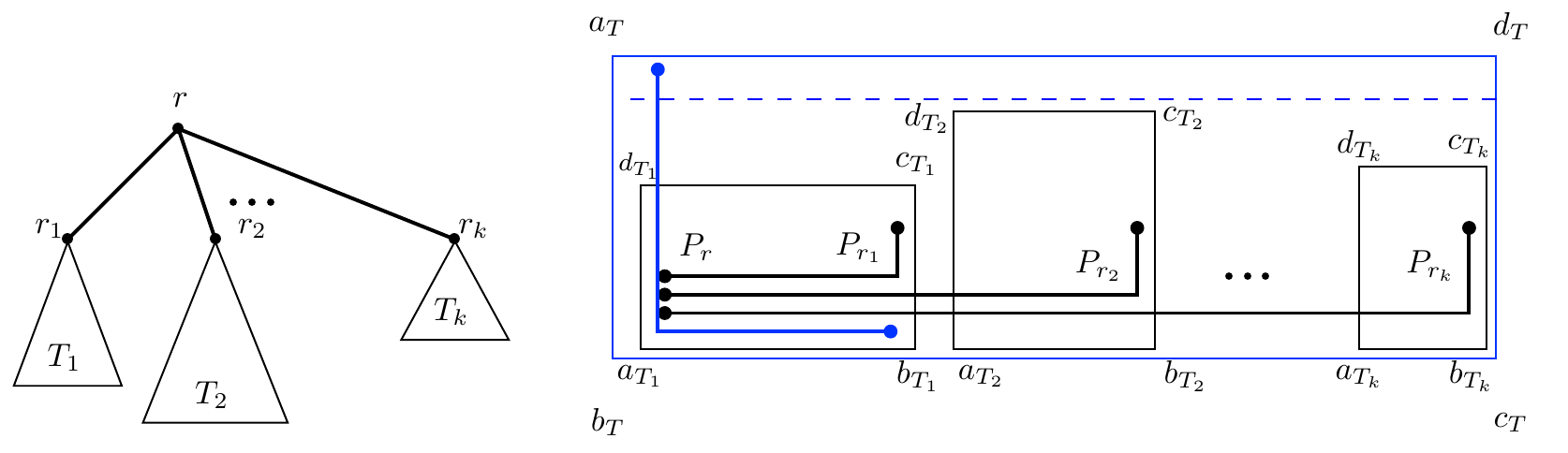}}
\caption{A construction for $\boneenpg$ representation of trees.}
\label{fig:B1ENPG-tree}
\end{figure}

\end{enumerate}
\end{proof}

\section{Split Graphs}\label{sec:split}
In this section, we present a characterization theorem (Theorem~\ref{thm:b1-enpgSplitStructure}) for $\boneenpg$ split graphs. In Sections \ref{subsec:split-maintheorem} and \ref{subsec:split-results} we proceed with some properties of these graphs implied by this theorem. An interesting implication of one of these properties is that the family of $\boneenpg$ is properly included in the family of $\btwoenpg$ graphs. Finally, using Theorem~\ref{thm:b1-enpgSplitStructure}, we prove in Section~\ref{subsec:split-npc} that the recognition problem of  $\boneenpg$ graphs is $\npc$ even in a very restricted subfamily of split graphs. Throughout this section, $G$ is a split graph $S(K,S,E)$ unless indicated otherwise. We assume without loss of generality that $K$ is maximal, i.e. that no vertex in $S$ is adjacent to all vertices of $K$, and $G$ is connected (in particular $S$ does not contain isolated vertices).

\subsection{Characterization of $\boneenpg$ Split Graphs}\label{subsec:split-maintheorem}
Consult Figure \ref{fig:B1ENPG-split-CharacterizationOnlyIf} for the following discussion.

\begin{figure}
\commentfig{\includegraphics[width=1.0\textwidth]{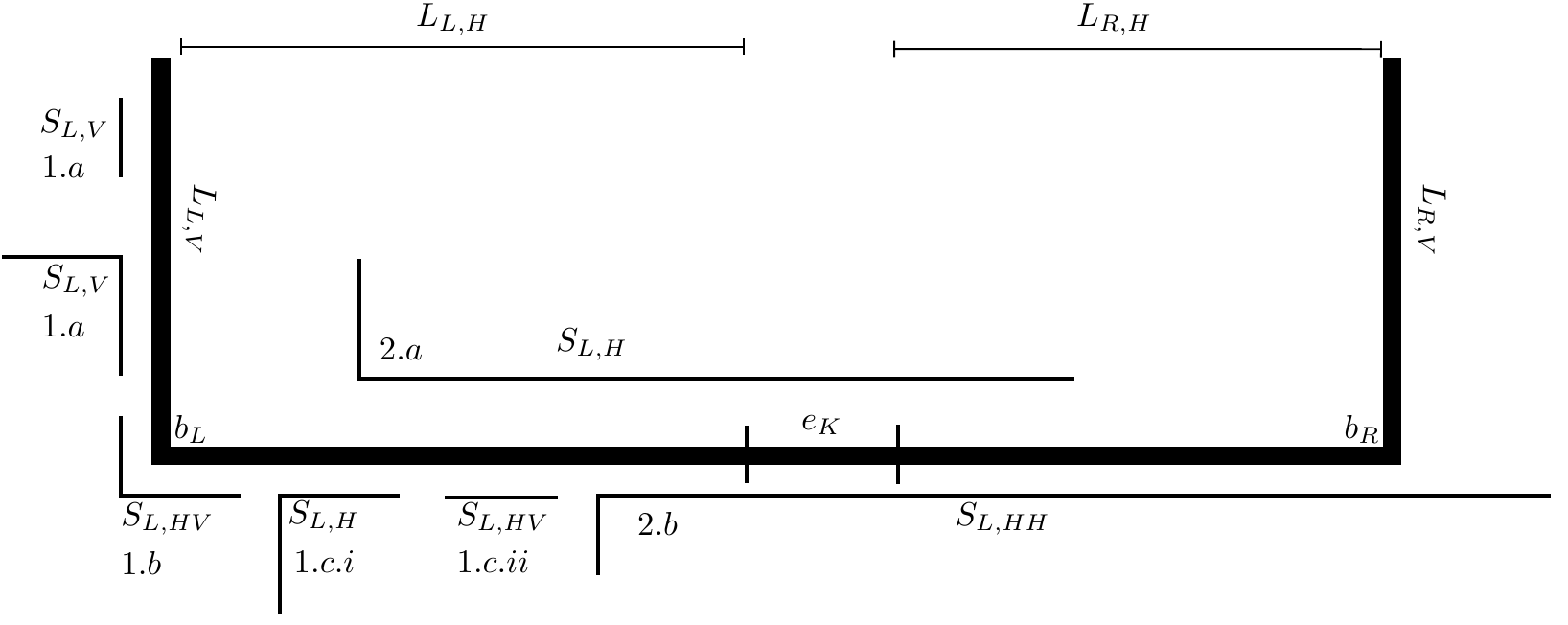}}
\caption{Regions of $\cup \pp_K$ and possible representations of a vertex in $S$ for a $\boneenpg$ split graph $S(K,S,E)$.}\label{fig:B1ENPG-split-CharacterizationOnlyIf}
\end{figure}

Let $G$ be a $\boneenpg$ split graph $S(K,S,E)$ with a representation $\rep$. By Proposition~\ref{prop:BendsInClique}, we know that $\cup \pp_K$ is a path with at most two bends, such that there is an edge $e_K$ contained in every path of $\pp_K$.
Moreover, if $\cup \pp_K$ contains two bends then $e_K$ is between the two bends.
Assume without loss of generality that $e_K$ is a horizontal edge.
Therefore, $\cup \pp_K$ consists of a horizontal segment $L_H$ between two vertices $b_L, b_R$ and two vertical segments (each of which is possibly empty). The subgrid $L_H \setminus e_K$ consists of two horizontal subsegments $L_{L,H}, L_{R,H}$. Finally, $\cup \pp_K \setminus L_H$ consists of two vertical segments $L_{L,V}$ and $L_{R,V}$. The segment $L_{L,Y}$ (resp. $L_{R,Y}$) is on the left (resp. right) of $e_K$ for every $Y \in \set{H,V}$.

For $(X,Y) \in \set{L,R} \times \set{H,V}$, let $K_{X,Y}$ be the set of vertices $v$ of $K$ such that $P_v$ has an endpoint in $L_{X,Y}$.
Every path of $\pp_K$ has its left (resp. right) endpoint on $L_{L,H} \cup L_{L,V}$ (resp. $L_{R,H} \cup L_{R,V}$) since it contains $e_K$.
Therefore, $\set{K_{X,H},K_{X,V}}$ is a partition of $K$ for every $X \in \set{L,R}$. For every $(X,Y) \in \set{L,R} \times \set{H,V}$, let $\sigma_{X,Y}$ be the permutation of $K_{X,Y}$ obtained by ordering the endpoints of  $\pp_K$  in $L_{X,Y}$ in increasing distance from $e_K$.
Moreover, $K_{L,V} \cap K_{R,V} = \emptyset$ since otherwise this implies a path containing both $b_L$ and $b_R$ as bends.

The following theorem characterizes the $\boneenpg$ split graphs. If further  provides a canonical representation for them using the above mentioned partitions and by partitioning the vertices of $S$ according their neighborhoods.

\begin{theorem}\label{thm:b1-enpgSplitStructure}
A connected split graph $G=S(K,S,E)$ is $\boneenpg$ if and only if there are two partitions $\set{K_{L,H},K_{L,V}}$, $\set{K_{R,H},K_{R,V}}$ of $K$ such that $K_{L,V} \cap K_{R,V}=\emptyset$, there is a permutation $\sigma_{X,Y}$ of $K_{X,Y}$ for every $(X,Y) \in \set{L,R} \times \set{H,V}$, and a partition $\cs = \set{S_{X,H},S_{X,V},S_{X,HH},S_{X,HV} | X \in \set{L,R}}$ of $S$ such that the following hold.
\begin{enumerate}[i)]
\item {}\label{itm:Interval} If $s \in S_{X,Y}$ then $N(s)$ is an interval $\sigma_s$ of $\sigma_{X,Y}$.
\item {}\label{itm:PrefixIntersection} If $s \in S_{X,HH}$ then $N(s)$ consists of the intersection of a prefix $\sigma_s$ of $\sigma_{X,H}$ with $K_{\bar{X},H}$ where $\bar{X} = \set{L, R} \setminus X$.
\item {}\label{itm:SuffixUnion} If $s \in S_{X,HV}$ then $N(s)$ is the union of a suffix $\sigma_s$ of $\sigma_{X,H}$ with $K_{X,V}$.
\item {}\label{itm:SuffixPrefix} If $s \in S_{X,H} \cup S_{X,HH}$ then there is at most one $s' \in S_{X,HV}$ such that the interval $\sigma_s$ of $\sigma_{X,H}$ and the suffix $\sigma_{s'}$ of $\sigma_{X,H}$ overlap.
\item {}\label{itm:UniqSuffixUnion} If $S_{{X},HH} \neq \emptyset$ then $\abs{S_{\bar{X,}HV}} \leq 1$ where $\bar{X} = \set{L, R} \setminus X$.
\item {}\label{itm:VerticalSmallerThanOtherHorizontal} $K_{X,V} \subseteq K_{\bar{X},H}$  where $\bar{X} = \set{L, R} \setminus X$.
\end{enumerate}
\end{theorem}
\begin{proof}
$(\Rightarrow)$
We fix a $\boneenpg$ representation of $G$ and consider the sets $K_{X,Y}$ and their permutations $\sigma_{X,Y}$ defined by this representation.\\
\ref{itm:Interval}, \ref{itm:PrefixIntersection} \ref{itm:SuffixUnion}) For each vertex $s \in S$ we will determine its membership to one of the sets of the partition $\cs$ depending on its representation $P_s$.
Suppose that there exists a vertex $s \in S$ such that $\abs{\cs(P_s,\cup \pp_K)} > 1$.
Then $P_s \cup \cup \pp_K$ contains a cycle, therefore at least $4$ bends.
But $P_s$ has at most one bend and $\cup \pp_K$ has at most two bends, a contradiction.
Therefore, $\cs(P_s,\cup \pp_K)$ consists of one segment $Q_s$.

Given a vertex $s \in S$, we consider two disjoint and complementary cases for $P_s$.
\begin{enumerate}
\item{$e_K \notin Q_s$:}
Let $c_s$ (resp. $f_s$) be the vertex of $Q_s$ closer to (resp. farther from) $e_K$.
Assume without loss of generality that $Q_s \subseteq L_L$ where $L_L = L_{L,H} \cup L_{L,V}$ (we define similarly $L_R = L_{R,H} \cup L_{R,V}$).
We observe that $P_s$ does not split from $L_L$ at $c_s$, since otherwise $P_s$ splits from every path of $\pp_K$ that it intersects, implying that $s$ is an isolated vertex. We further consider three subcases:
\begin{enumerate}
\item{$Q_s \subseteq L_{L,V}$:} If $P_s$ splits from $L_L$ then $N(s)$ consists of the vertices $v$ of $K$ such that the left endpoint of $P_v$ is between $c_s$ and $f_s$.
    Therefore, $N(s)$ is an interval of $\sigma_{L,V}$.
    If $P_s$ does not split from $L_L$ then $N(s)$ consists of the vertices $v$ of $K$ such that the left endpoint of $P_v$ is farther than $c_s$ on $L_{L,V}$ (with respect to $b_L$).
    Therefore, $N(s)$ is an interval of $\sigma_{L,V}$.
    In both cases we set $s \in S_{L,V}$.
\item{$b_l$ is an internal vertex of $Q_s$:} In this case $P_s$ does not split from $L_L$. Then $N(s)$ consists of the vertices $v$ of $K$ such that the left endpoint of $P_v$ is on the left of $c_s$ on $L_L$. Therefore, $N(s)$ consists of the union of a suffix of $\sigma_{L,H}$ with $K_{L,V}$, and we set $s \in S_{L,HV}$.
\item{$Q_s \subseteq L_{L,H}$:}
\begin{enumerate}
\item{$P_s$ splits from $L_L$.} In this case $N(s)$ consists of the vertices $v$ of $K$ such that the left endpoint of $P_v$ is between $c_s$ and $f_s$.
    Therefore, $N(s)$ is an interval of $\sigma_{L,H}$ and we set $s \in S_{L,H}$.
\item{$P_s$ does not split from $L_L$.} In this case $N(s)$ consists of the vertices $v$ of $K$ such that the left endpoint of $P_v$ is on the left of $c_s$ on $L_L$.
    Therefore, $N(s)$ consists of the union of a suffix of $\sigma_{L,H}$ with $K_{L,V}$ and we set $s \in S_{L,HV}$.
\end{enumerate}
\end{enumerate}
\item{$e_K \in Q_s$:} In this case, let $l_s, r_s$ be the endpoints of $Q_s$ on $L_L$ and $L_R$ respectively. The path $P_s$ splits from $\cup \pp_K$, since otherwise $s$ is adjacent to every vertex of the clique, contradicting the maximality of $K$. Since $P_s$ intersects every path of $\pp_K$, $N(s)$ consists of the vertices $v$ of $K$ such that $P_v$ does not split from $P_s$. We consider two subcases:
\begin{enumerate}
\item{$P_s$ splits from exactly one of $L_L$ and $L_R$:} Assume without loss of generality that $P_s$ splits from $L_L$ but not from $L_R$. We observe that $l_s \in L_{L,H}$. In this case $N(s)$ is the set of vertices $v$ of $K$ such that the left endpoint of $P_v$ is closer than $l_s$ to $e_K$ which corresponds to a prefix of $\sigma_{L,H}$. In this case we set $s \in S_{L,H}$.
\item{$P_s$ splits from both of $L_L, L_R$:} In this case at least one of the endpoints of $Q_s$ is a bend of $\cup \pp_K$, i.e., $\set{l_s, r_s} \cap \set{b_L, b_R} \neq \emptyset$.
    Assume without loss of generality that $r_s=b_R$.
    Then $N(s)$ consists of those vertices $v$ of $K$ such that the left endpoint of $P_v$ is closer to $e_K$ than $l_s$ and the right endpoint of $P_v$ is in $L_{R,H}$.
    This is exactly a prefix of $\sigma_{L,H}$ intersected with $K_{R,H}$, thus we set $s \in S_{L,HH}$.
\end{enumerate}
\end{enumerate}

\itemref{SuffixPrefix} Assume, for a contradiction that for some $X \in \set{L,R}$, say $L$, the condition does not hold, i.e. there is a vertex $s \in S_{L,H} \cup S_{L,HH}$ and two vertices $s',s'' \in S_{L,HV}$ such that the interval $\sigma_s$ of $\sigma_{L,H}$ corresponding to $s$ overlaps both of the suffixes $\sigma_{s'}, \sigma_{s''}$ corresponding to $s', s''$ respectively. By the above case analysis we know that $P_s$ is either of type 2 or of type 1.c.i, and that $P_{s'}$ and $P_{s''}$ are of one of the types 1.b, 1.c.ii. Note that $\sigma_{s'}$ and $\sigma_{s'}$ are determined by the right endpoints of the corresponding paths $P_{s'}$ and $P_{s''}$. Since $\sigma_{s'}$ and $\sigma_s$ overlap, $Q_s$ contains the right endpoint of $P_{s'}$. Therefore these two paths intersect.  Moreover $P_{s'}$ contains the left endpoint of $Q_s$ (which is the bend point of $P_s$) otherwise $P_s \sim P_{s'}$. By the same arguments $P_{s''}$ also contains the left endpoint of $Q_s$ and therefore $P_{s'}$ intersects $P_{s''}$. Moreover, since none of them splits from $L_L$ we have $P_{s'} \sim P_{s''}$, i.e. $s$ and $s'$ are adjacent in $G$, a contradiction.

\itemref{UniqSuffixUnion} Assume, for a contradiction that for some $X \in \set{L,R}$, say $L$, the condition does not hold, i.e. there exists $s \in S_{R,HH}$ and $s',s'' \in S_{L,HV}$.
    Then both $P_{s'}$ and $P_{s''}$ are of one of the types 1.b, 1.c.ii. Moreover, $P_s$ is of type 2.b with $l_s = b_L$. We proceed as in the previous case to get a contradiction.

\itemref{VerticalSmallerThanOtherHorizontal} This immediately follow from the already observed fact that $K_{L,V} \cap K_{R,V}=\emptyset$.

$(\Leftarrow)$
Suppose that the partitions and the permutations stated in the claim exist.
We construct a representation $\rep$ as follows (see Figure \ref{fig:B1ENPG-split-CharacterizationIf}). The host graph $H$ is a $(2+4 \abs{S} (\abs{K}+1))$ by $2 \abs{S} \cdot \abs{K}$ vertices grid, where each vertex is represented by an ordered pair from $[-1-2 \abs{S} (\abs{K}+1),1+2 \abs{S} (\abs{K}+1)] \times [0,2 \abs{S} \abs{K}]$.
For $(X,Y) \in \set{L,R} \times \set{H,V}$, let $k_{X,Y} = \abs{K_{X,Y}}$. The coordinates of $b_L$ and $b_R$ are respectively $(-1 - 2\abs{S} (k_{L,H}+1), 0)$ and $(1 + 2\abs{S} (k_{L,H}+1),0)$. The horizontal line between $b_L$ and $b_R$ is called $L_H$. For $(X,Y) \in \set{L,R}$, $L_{X,V}$ is a vertical line of length 2$\abs{S} (k_{X,V})$ starting at $b_X$. We choose $k_{X,Y}$ vertices on each line $L_{X,Y}$ such that their distances from each other and from each of $b_L, b_R, (-1,0), (1,0)$ is at least $2\abs{S}$.
We label these vertices as $w_{X,\sigma_{X,Y}(1)}, \ldots, w_{X,\sigma_{X,Y}(k_{X,Y})}$ in increasing order of their distances from the origin. Every vertex $v \in K$ is represented by a path $P_v \subseteq \cup L_H \cup L_{L,V} \cup L_{R,V}$ between $w_{L,v}$ and $w_{R,v}$. Since $K_{X,V} \subseteq K_{\bar{X},H}$ every such path has at most one bend. Since $e_K=(-1,0)(1,0)$ is contained in every such path, these paths constitute a proper representation of the clique $K$.

We proceed with the representation of the vertices of $S$. Let $W_X = \set{w_{X,1}, \ldots w_{X,\abs{K}}}$. The endpoints of paths $Q_s, s \in S$ will be chosen between two vertices of $W \cup \set{b_X}$ so that they are all distinct. We first determine the representations of the vertices of $S \setminus \set{S_{L,HV} \cup S_{R,HV}}$ such that all are one bend paths with distinct bends on the endpoint of $Q_s$ farther from the origin. We determine the endpoints of $Q_s$ according to the permutation $\sigma_s$ and as close to the origin as possible. Since all these paths have distinct bends, they represent an independent set. Last, we represent the vertices $s' \in \set{S_{L,HV} \cup S_{R,HV}}$. For each such vertex its representation will be a path $P_{s'} \subseteq \cup \pp_K$. We choose the endpoint closer to the origin according to the suffix $\sigma_{s'}$. Let $O_{s'}$ be the set of vertices such that for all vertices $s \in O_{s'}$, $\sigma_s$ and $\sigma_{s'}$ overlap. The other endpoint is chosen as the endpoint closest to the origin that is farther from the origin than all the endpoints of the paths $P_s$ where $s \in O_{s'}$. Conditions \ref{itm:SuffixPrefix} and \ref{itm:UniqSuffixUnion} guarantee that after this is done, every path $P_{s'}$ that intersects with $P_s$ splits from it.
\begin{figure}
\commentfig{\includegraphics[width=1.0\textwidth]{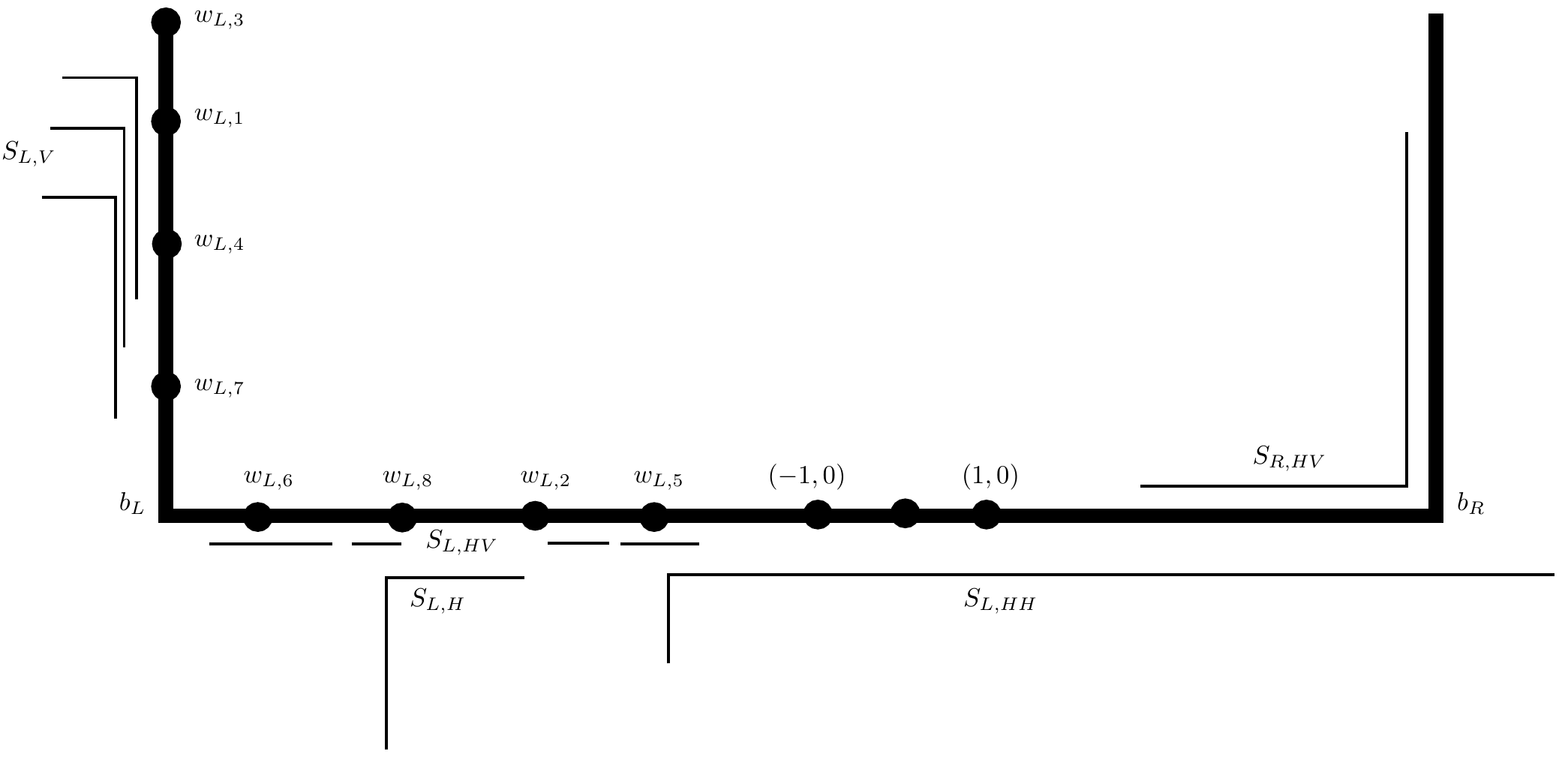}}
\caption{Construction of a representation $\rep$ for a split graph satisfying conditions i)--vi) of Theorem \ref{thm:b1-enpgSplitStructure}.}\label{fig:B1ENPG-split-CharacterizationIf}
\end{figure}
\end{proof}

From the proof of Theorem \ref{thm:b1-enpgSplitStructure} we obtain the following corollary where a \emph{caterpillar} is a tree in which all vertices are within distance 1 of a central path.
\begin{corollary}\label{coro:B1ENPGSplitSubsetB1ENPTSplit}
~\\
i) Every connected $\boneenpg$ split graph $G=S(K,S,E)$ has a representation $\rep$ such that $\cup \pp$ is a caterpillar with central path $\cup \pp_K$ and maximum degree 3.\\
ii) $\boneenpg \cap \textsc{SPLIT} \subseteq \enpt \cap \textsc{SPLIT}.$
\end{corollary}

\subsection{Two Consequences of The Characterization of $\boneenpg$ Split Graphs}\label{subsec:split-results}
Throughout this section, we use the notation introduced in the previous section.
$X$ (resp. $Y$) denotes an element of $\set{L,R}$ (resp. $\set{H,V}$), and $\bar{X}=\set{L,R} \setminus \set{X}$.
Given a $\boneenpg$ split graph $G=S(K,S,E)$, we denote by $K_{X,Y}, S_{X,Y}, S_{X,HV}, S_{X,HH}$ and $\sigma_{X,Y}$ the sets and permutations whose existence are guaranteed by Theorem \ref{thm:b1-enpgSplitStructure}.
Furthermore $\sigma_X = \sigma_{X,H} \cdot \sigma_{X,V}$ is the permutation of $K$ obtained by the concatenation of the two permutations $\sigma_{X,H}$ and $\sigma_{X,V}$. We define $S_{X,Y,d}$ as the set of vertices of $S_{X,Y}$ having degree $d$ in $G$. The notations $S_{X,HV,d}$, and $S_{X,HH,d}$ are defined similarly.

The following inequalities are easy to show using the definitions of the sets and counting the number of prefixes, suffixes, or intervals of a given permutation having a given length.
\begin{proposition}\label{prop:SplitGraphSetSizes}
If $G=S(K,S,E)$ is a twin-free and (false twin)-free $\boneenpg$ split graph, then
\begin{eqnarray}
\abs{S_{X,Y,d}} & \leq & \max \set{\abs{K_{X,Y}}+1-d,0} \label{eqn:SXYd}\\
\abs{S_{X,HV,d}} & \leq & \left\{ \begin{array}{ll}
        1 & \textrm{if~} d > \abs{K_{X,V}}\\
        0 & \textrm{otherwise}
    \end{array}
    \right.\label{eqn:SXHVd} \\
\abs{S_{X,V,d} \cup S_{X,HV,d}} & \leq & \left\{ \begin{array}{ll}
        1 & \textrm{if~} d > \abs{K_{X,V}}\\
        \abs{K_{X,V}}+1-d & \textrm{otherwise}
    \end{array}
    \right.\label{eqn:SXHVVd} \\
\abs{S_{X,HH,d}} & \leq & \left\{ \begin{array}{ll}
        1 & \textrm{if~} d \leq \abs{K_{L,H} \cap K_{R,H}}\\
        0 & \textrm{otherwise}
    \end{array}
    \right. \leq 1\label{eqn:SXHHd} \\
\abs{S_d} & \leq & 2 \left( \abs{K} + 2 - d \right). \label{eqn:Sd}
\end{eqnarray}
\end{proposition}
\begin{proof}
(\ref{eqn:SXYd}) If $s \in S_{X,Y,d}$ then $N(s)$ is an interval $\sigma_s$ of $\sigma_{X,Y}$ of size $d$. Since $G$ is (false twin)-free $N(s) \neq N(s')$ whenever $s \neq s'$. Therefore, $\abs{S_{X,Y,d}}$ is at most the number of such intervals which is given by the right hand size of the inequality.

(\ref{eqn:SXHVd}) If $s \in S_{X,HV,d}$ then $N(s) = K_{X,V} \cup \sigma_s$ where $\sigma_s$ is a suffix of $\sigma_{X,H}$. Therefore, $d > \abs{K_{X,V}}$ and $\sigma_s$ is the unique suffix of $\sigma_{X,H}$ of size $d-\abs{K_{X,V}}$.

(\ref{eqn:SXHVVd}) Follows from (\ref{eqn:SXYd}) and (\ref{eqn:SXHVd}).

(\ref{eqn:SXHHd}) If $s \in S_{X,HH,d}$ then $N(s) = K_{\bar{X},H} \cap \sigma_s$ where $\sigma_s$ is a prefix of $\sigma_{X,H}$ of size $d$. If $d > \abs{K_{L,H} \cap K_{R,H}}$ no such prefix exists, otherwise there is exactly one such prefix.

(\ref{eqn:Sd}) By summing up (\ref{eqn:SXHVd}), (\ref{eqn:SXHVVd}), and also (\ref{eqn:SXYd}) for $Y=H$ and finally multiplying by two for the two possible values of $X$.
\end{proof}

Summing up (\ref{eqn:Sd}) for all the possible values of $d \in [\abs{K}]$ we get the following corollary.
\begin{corollary}
If $G=S(K,S,E)$ is a (false twin)-free $\boneenpg$ split graph, then
$\abs{S}$ is $\bigoh (\abs{K}^2)$.
\end{corollary}
Using similar arguments one can show that if $G=S(K,S,E)$ is twin-free then $\abs{K}$ is $\bigoh (\abs{S}^2)$ implying that $\abs{S}$ is $\Omega (\sqrt{\abs{K}})$. More specifically, one should consider the set of endpoints or bend points of the paths $\pp_S$ all of which are in $\cup \pp_K$ and observe that no four such points that are pairwise consecutive in $\cup \pp_K$ may surround the endpoints of two paths $P_u, P_v \in \pp_K$ since otherwise $u$ and $v$ are twins.

\begin{theorem}\label{thm:BoneProperlyIncludedInBTwo} The following strict inclusions hold:
\begin{itemize}
\item $\boneenpg \cap \textsc{SPLIT} \subsetneq \benpg{2} \cap \textsc{SPLIT}.$
\item $\boneenpg \cap \textsc{SPLIT} \subsetneq \enpt \cap \textsc{SPLIT}.$
\end{itemize}
\end{theorem}
\begin{proof}
By the definition of a $\benpg{k}$ graph, we have $\boneenpg \cap \textsc{SPLIT} \subseteq \benpg{2} \cap \textsc{SPLIT}$ and by Corollary \ref{coro:B1ENPGSplitSubsetB1ENPTSplit}, we have $\boneenpg \cap \textsc{SPLIT} \subseteq \enpt \cap \textsc{SPLIT}$. In the following, we provide a split graph with a representation which is both $\benpg{2}$ and $\enpt$. We show that this split graph is not $\boneenpg$.

Let $K=[0,10]$, $\sigma_L$ be the identity permutation $(0,1,2,3,4,5,6,7,8,9,10)$ on $K$ and $\sigma_R$ be the permutation $(0,5,10,4,9,3,8,2,7,1,6)$ on $K$. Let $G=S(K,S,E)$ where $S$ contains 23 vertices: one for every pair that is consecutive in one of $\sigma_L, \sigma_R$ (there are 10 in every permutation and we note that these pairs are distinct) and one for each of the pairs $\set{0,2},\set{0,3},\set{0,4}$ (which are not consecutive in any of these permutations). Every vertex in $S$ is adjacent to the corresponding pair in $K$. Note that $G$ is (false-twin)-free and $\abs{S_2}=\abs{S}=23>22=2(\abs{K}+2-2)$. By Proposition \ref{prop:SplitGraphSetSizes} (\ref{eqn:Sd}), $G$ is not $\boneenpg$. Figure \ref{fig:B2ENPGTree} depicts a set of paths that constitute a $\benpg{2}$ representation and an $\enpt$ representation of $G$.
\begin{figure}
\centering
\commentfig{\includegraphics{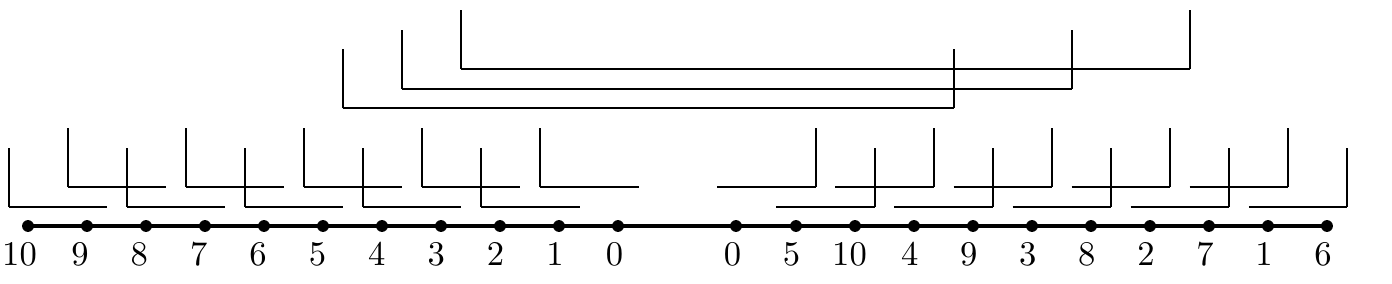}}
\caption{The $\benpg{2}$ representation of a non-$\boneenpg$ split graph used in the proof of Theorem~\ref{thm:BoneProperlyIncludedInBTwo}. The paths representing vertices of $K$ are not drawn. They are implied by the numbers of the vertices: for every $i \in [0,10]$, $P_i$ is the shortest path between the two vertices labeled $i$. The paths with two bends, intersects every path of $\pp_K$, but splits from every path having an endpoint after its bend points. Therefore, these path represent vertices with neighborhood $\set{0,2}$, $\set{0,3}$, $\set{0,4}$.}. \label{fig:B2ENPGTree}
\end{figure}
\end{proof}

\subsection{NP-completeness of $\boneenpg$ split graph recognition}\label{subsec:split-npc}
We now proceed with the $\textsc{NP}$-completeness of $\boneenpg$ recognition in split graphs. We first present a preliminary result that can be useful per se.

A graph is $d$-regular if all its vertices have degree $d$. A $3$-regular graph is also termed \emph{cubic}.
A \emph{diamond} is the graph $K_4 - e$ obtained by removing an edge from a clique on four vertices.
Clearly, if the edge set of a graph $G$ can be partitioned into two Hamiltonian cycles, then $G$ is $4$-regular. However, in the opposite direction we have the following:

\begin{theorem}\label{thm:decompositionOfFourRegular}
The problem of determining whether the edge set of a diamond-free $4$-regular graph can be partitioned into two Hamiltonian cycles is $\npc$.
\end{theorem}
\begin{proof}
We prove by reduction from the Hamiltonicity problem of cubic bipartite graphs which is known to be $\npc$ (\cite{ANS80-HamiltonicityofBipartiteCubic}.
Let $G$ be a cubic bipartite graph whose Hamiltonicity has to be decided, and let $H=L(G)$ be its line graph.
$H$ is clearly $4$-regular. In the sequel we will show that $H$ is also diamond-free.
In addition, we know that $G$ is Hamiltonian if and only if the edge set of its line graph $H$ can be partitioned into two Hamiltonian cycles (\cite{K57}).
This concludes the proof. It remains to show that $H$ is diamond-free.

Suppose for a contradiction that $H$ contains a diamond on vertices $\set{e_1,e_2,e_3,e_4}$ that are pairwise adjacent except for the pair $e_1, e_4$.
Then $\set{e_1,e_2,e_3}$ and $\set{e_2,e_3,e_4}$ are two triangles of $H$.
Every triangle of $H$ corresponds to either a triangle of $G$, or to three edges of $G$ incident to a common vertex. Since $G$ is bipartite, only the latter case is possible.
Then $e_1,e_2,e_3$ (resp. $e_2,e_3,e_4$) are edges of $G$ incident to a vertex $v$ (resp. $v'$).
Since $G$ is cubic we have $v \neq v'$.
We conclude that $e_2 = e_3 = vv'$, a contradiction.
\end{proof}

We are now ready to prove the main result of this section.
\begin{theorem}
The recognition problem of $\boneenpg$ graphs is $\npc$ even when restricted to split graphs.
\end{theorem}
\begin{proof}
The proof is by reduction from the problem of decomposing a $4$-regular, diamond-free graph into two Hamiltonian cycles which is shown to be $\npc$ in Theorem \ref{thm:decompositionOfFourRegular}.
Given a $4$-regular graph $G$ on $n+1$ vertices, we remove an arbitrary vertex $v$ of $G$ and obtain the graph $G'=G-v$ on $n$ vertices all of which having degree 4, except the four neighbours $\set{v_1,v_2,v_3,v_4}$ of $v$ each of which having degree 3.
We construct the split graph $G''=S(K,S,E)$ where $K=V(G')$, $S=E(G') \cup \set{s_1,s_2,s_3,s_4}$.
Furthermore, the neighborhood of a vertex $s \in S$ is determined as follows.
If $s=s_i$ for some $i \in [4]$ then $N_{G''}(s) = K - v_i$, otherwise $s$ is an edge $uv$ of $G'$ in which case $N_{G''}(s) = \set{u,v}$.
It remains to show that $G''$ is $\boneenpg$ if and only if $E(G)$ can be partitioned into two Hamiltonian cycles.

Assume that $E(G)$  can be partitioned into two Hamiltonian cycles $C_L, C_R$.
This induces a partition of $E(G')$ into two paths $Q_L$ and $Q_R$
which in turn induces a partition of $S-\set{s_1,s_2,s_3,s_4}$ into $S_{L,H}=E(Q_L)$ and $S_{R,H}=E(Q_R)$.
Note that the endpoints of $Q_L$ and $Q_R$ are the degree $3$ vertices of $G'$, i.e. $\set{v_1, v_2, v_3, v_4}$.
Let without loss of generality $v_1, v_2$ (resp. $v_3,v_4$) be the endpoints of $Q_L$ (resp. $Q_R$).
For $X \in \set{L,R}$ let $K_{X,H}=K=V(G')$ and $K_{X,V}=\emptyset$.
We set $\sigma_{X,H}$ as the order of the vertices of $G'$ in $Q_X$ (which is a permutation of the vertices of $K=V(G')$).
Then $v_1$ and $v_2$ (resp. $v_3$ and $v_4$) are the first and last vertices of the permutation $\sigma_{L,H}$ (resp. $\sigma_{R,H}$).
For $X \in \set{L,R}$ we set $S_{X,V}=S_{X,HH}=S_{X,HV}=\emptyset$.
We now verify that these settings satisfy the conditions of Theorem~\ref{thm:b1-enpgSplitStructure}. Conditions \itemref{PrefixIntersection}, \itemref{SuffixUnion}, \itemref{SuffixPrefix}, \itemref{UniqSuffixUnion} and \itemref{VerticalSmallerThanOtherHorizontal} easily follow since the sets $S_{X,V}$, $S_{X,HV}$, $S_{X,HH}$ and $K_{X,V}$ are empty. As for Condition \itemref{Interval} we consider two cases.
If $s=s_i$ for some $i \in [4]$ then $N_{G''}(s)=K-v_i$ is an interval of $\sigma_{X,H}$ for some $X \in \set{L,R}$ since $v_i$ is either the first or the last vertex of one of these permutations.
If $s$ is an edge $uv \in E(Q_X)$ of $G'$ then $u$ and $v$ are consecutive in the permutation $\sigma_{X,H}$.
Since all the conditions are satisfied, we conclude that $G''$ is $\boneenpg$.

Now assume that $G''$ is $\boneenpg$.
For $X \in \set{L,R}, Y \in \set{H,V}$, let $K_{X,Y}$, $\sigma_{X,Y}$, $S_{X,Y}$, $S_{X,HV}$ and $S_{X,HH}$ be sets and permutations whose existence are guaranteed by Theorem~\ref{thm:b1-enpgSplitStructure}.
We first show that $\abs{K_{X,V}} \leq 1$.
Assume for a contradiction that $\abs{K_{X,V}} > 1$ for some $X \in \set{L,R}$, say $X=L$.
Then we have $\abs{K_{R,H}} > 1$, and $\abs{K_{L,H}}, \abs{K_{R,V}} < n-1$.
By Proposition \ref{prop:SplitGraphSetSizes}, these imply $S_{L,H,n-1}=S_{L,HH,n-1}=S_{R,HH,n-1}=\emptyset$.
Moreover, $\abs{S_{R,H,n-1}} \leq 2$ and this may hold with equality only when $K_{R,H}=K$.
Finally, we have $\abs{S_{X,V,n-1} \cup S_{X,HV,n-1}} \leq 1$.
Summing up all inequalities we obtain $\abs{S_{n-1}} \leq 4$.
We recall that all the vertices of $S$ have degree $2$ except the four special vertices with degree $n-1$.
Therefore, $S_{n-1} = \set{s_1,s_2,s_3,s_4}$, and we conclude that all the inequalities hold with equality.
In particular $\abs{S_{R,H,n-1}}=2$, implying $K_{R,H}=K$ and $K_{R,V}=\emptyset$.
Then we have $S_{R,V,n-1} \cup S_{R,HV,n-1} = \emptyset$, i.e. one of the inequalities is strict, a contradiction.
Therefore, $\abs{K_{X,V}} \leq 1$, implying
\begin{equation}\label{eqn:SplitNPCompletenessOne}
S_{X,V,2} = \emptyset.
\end{equation}

Recall that $\sigma_X = \sigma_{X,H} \cdot \sigma_{X,V}$.
We now show that the set of the first and last vertices of $\sigma_L$ and $\sigma_R$ is $\set{v_1,v_2,v_3,v_4}$.
Let $i \in [4]$ and consider each one of the cases $s_i \in S_{X,H}$, $s_i \in S_{X,HH}$ and $s_i \in S_{X,HV}$ (the case $s_i \in S_{X,V}$ is impossible since $\abs{K_{X,V}} \leq 1$).
It is easy to verify for every case that $v_i$ is either the first or the last vertex of $\sigma_X$.
By the pigeonhole principle we conclude that the set of first and last vertices of $\sigma_L$ and $\sigma_R$ is $\set{v_1,v_2,v_3,v_4}$.
We assume without loss of generality that $v_1$ (resp. $v_2)$ is the first (resp. last) vertex of $\sigma_L$ and that $v_3$ (resp. $v_4)$ is the first (resp. last) vertex of $\sigma_R$.

Our next step is to show that $S_{X,HH,2}=\emptyset$.
Assume for a contradiction that this does not hold, and let $s$ be a vertex (without loss of generality) of $S_{L,HH,2}$.
Then $\sigma_s$ is a prefix with two vertices of $\sigma_{L,H} \cap K_{R,H} = \sigma_{L,H} \setminus K_{R,V} = \sigma_{L,H} - v_4$.
Clearly, the first vertex of $\sigma_s$ is $v_1$.
If the second vertex $w$ of $\sigma_{L,H}$ is not $v_4$, then $\sigma_s = v_1 w$ is an interval of $\sigma_{X,H}$ implying that $s\in S_{L,H}$, a contradiction.
Therefore, $w=v_4$, and $\sigma_s = v_1 x$ where $x$ is the third vertex of $\sigma_{L,H}$.
We conclude that $S_{L,HH,2}=\set{v_1 x}$.
Recall that $v_4$ has three incident edges in $G'$.
Since $v_4$ is the leftmost vertex of $\sigma_R$,
none of these edges is in $S_{R,HH} \cup S_{L,VH}$.
Moreover, at most one of them is in $S_{R,H} \cup S_{R,HV}$.
Therefore, at least two of them are in $S_{L,H}$.
Then these edges are necessarily $v_4 v_1$ and $v_4 x$.
We conclude that $\set{v_1, v_4, x}$ induces a triangle in $G'$.
In other words $\set{v_1, v_4, x, v}$ induces a diamond on $G$, a contradiction. Therefore
\begin{equation}\label{eqn:SplitNPCompletenessTwo}
S_{X,HH,2}=\emptyset.
\end{equation}

Finally, if $K_{X,V}=\emptyset$ we have $S_{X,HV}=\emptyset$ and $\abs{S_{X,H,2}} \leq n-1$. Otherwise, $\abs{K_{X,V}}=1, \abs{K_{X,H}}=n-1$ and we have $\abs{S_{X,HV,2}} \leq 1$ and $\abs{S_{X,H,2}} \leq n-2$. In both cases we have
\begin{equation}\label{eqn:SplitNPCompletenessThree}
\abs{S_{X,H,2} \cup S_{X,HV,2}} \leq n-1.
\end{equation}

Combining (\ref{eqn:SplitNPCompletenessOne}), (\ref{eqn:SplitNPCompletenessTwo}) and (\ref{eqn:SplitNPCompletenessThree}) we obtain $\abs{S_2} \leq 2 (n-1)$.
Since $\abs{S_2}=\abs{E(G')}=\abs{E(G)}-4=2(n-1)$, all the inequalities must hold with equality, in particular
$\abs{S_{X,H,2} \cup S_{X,HV,2}} = n-1$.
Therefore, every two consecutive vertices in $\sigma_X$ are adjacent in $G'$. In other words, the permutation $\sigma_X$ corresponds to a path $Q_X$ of $G'$.
The endpoints of $Q_L$ (resp. $Q_R$) are $v_1$ and $v_2$ (resp. $v_3$ and $v_4$). Adding two edges incident to $v$ to each $Q_X$, we get two edge disjoint Hamiltonian cycles of $G$.
\end{proof}

\section{Co-bipartite Graphs}\label{sec:cobipartite}
In Section~\ref{subsec:cobipartite-maintheorem} we characterize $\boneenpg$ co-bipartite graphs. We show that there are two types of representations for $\boneenpg$ co-bipartite graphs. For each type of representation, we characterize their corresponding graphs. These characterizations imply a polynomial-time recognition algorithm. In Section~\ref{subsec:efficient-recognition} we present an efficient (linear-time) implementation of the algorithm.

\subsection{Characterization of $\boneenpg$ Co-bipartite Graphs}\label{subsec:cobipartite-maintheorem}

We proceed with definitions and two related lemmas (Lemma~\ref{lem:difference2meet}, Lemma~\ref{lem:meet2difference}) that will be used in each of the above mentioned characterizations.

Let $S$ be a path of a graph $H$ with endpoints $u,v$. Two sets $\pp_u$, $\pp_v$ of paths \emph{meet at} $S$ if for $x \in \set{u,v}$ (a) every path of $\pp_x$ contains $x$ (b) every path of $\pp_x$ has an endpoint that is a vertex of $S$ different than $x$, and (c) a pair of paths $P_u \in \pp_u$, $P_v \in\pp_v$ may intersect only in $S$ (see Figure~\ref{fig:TwoPathSetsMeet}).

\begin{figure}
\centering
\commentfig{\includegraphics{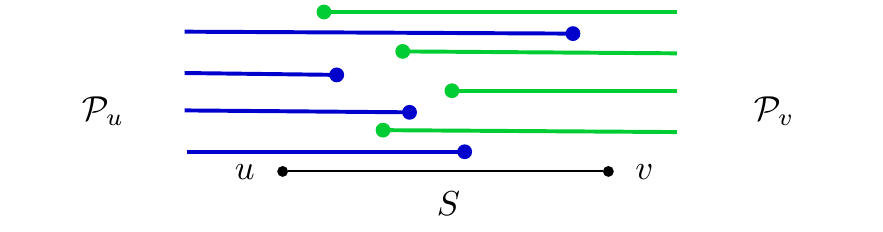}}
\caption{Two path sets $\pp_u$, $\pp_v$ meet at a path $S$ with endpoints $u$ and $v$.}
\label{fig:TwoPathSetsMeet}
\end{figure}

A graph $G=(V,E)$ is a \emph{difference graph} (equivalently \emph{bipartite chain graph}) if every $v_i \in V$ can be assigned a real number $a_i$ and there exists a positive real number $T$ such that (a) $\abs{a_i} < T$ for all $i$ and (b) $\set{v_i, v_j} \in E$ if and only if $\abs{a_i - a_j} \geq T$. Every difference graph is bipartite where the bipartition is according to the sign of $a_i$.

\begin{theorem}\label{theorem:difference-degrees}\cite{HPS90}
If $G = (V,E)$ is a bipartite graph with bipartition $V=X \cup Y$ then the following statements are equivalent:
\begin{enumerate}[i)]
\item $G$ is a difference graph.
\item Let $\delta_1 < \delta_2 < \ldots \delta_s$ be distinct nonzero degrees in $X$, and $\delta_0 = 0$. Let $\sigma_1 < \sigma_2 < \ldots \sigma_t$ be distinct nonzero degrees in $Y$, and $\sigma_0 = 0$. Let $X = X_0 \cup X_1 \cup \ldots X_s$, $Y = Y_0 \cup Y_1 \cup \ldots \cup Y_t$, where $X_i = \set{x \in X | d(x) = \delta_i}$, $Y_j = \set{y \in Y | d(y) = \delta_j}$. Then $s=t$ and for $x \in X_i$, $y\in Y_j$, $\set{x,y} \in E$ if and only if $i+j > t$.
\end{enumerate}
\end{theorem}

\begin{theorem}\label{theorem:difference-2K2free}\cite{HPS90}
A graph is a difference graph if and only if it is bipartite and $2K_2$-free.
\end{theorem}

\begin{lemma}\label{lem:difference2meet}
Let $G_B = B(K,K',E)$ a difference graph, and $t$ be the number of distinct nonzero degrees of vertices of $K$ in $G_B$.
Let $H$ be a grid and $S$ be a path of $H$ with length at least $t+2$ and no bends.
Then there is a $\boneenpg$ representation $\rep$ of $G_C=C(K,K',E)$ such that $\pp_K$ and $\pp_{K'}$ meet at $S$.
\end{lemma}

\begin{proof}
Let $\delta_1 < \delta_2 < \ldots \delta_s$ (resp. $\sigma_1 < \sigma_2 < \ldots \sigma_t$) be the distinct nonzero degrees in $K$ (resp in $K'$) in $G_B$. By Theorem~\ref{theorem:difference-degrees} we have $s=t$. Let $-1,0,1, \ldots,t+1$ be $t+3$ vertices of $S$ such that $0$ and $t+2$ are the endpoints of $S$ and they appear in this order on $S$. Let $x$ (resp. $x'$) be a vertex of $K$ (resp. $K'$), and let $i$ be such that $d_{G_B}(x)=\delta_i$ (resp. $(d_{G_B}(x')=\sigma_{i'})$). The path $P_x$ (resp. $P_{x'}$) is constructed between vertices $-1$ and $i$ (resp. $t-j$ and $t+1$).

With this construction $\pp_K, \pp_{K'}$ represent the cliques $K$ and $K'$, moreover they meet at $S$. By the construction two paths $P_x \in \pp_K$, $P_{x'} \in \pp_{K'}$ intersect if and only if $i+j>t$. By Theorem~\ref{theorem:difference-degrees} $x$ and $x'$ are adjacent if and only if $i+j>t$. Therefore, $\pp$ is a representation of $G=C(K,K',E)$.
\end{proof}

\begin{lemma}\label{lem:meet2difference}
~\\
i) If two sets $\pp_K, \pp_{K'}$ of one-bend paths meet at a path $S$ then $G_B = B(K,K',E)$ is a difference graph.\\
ii) If a cobipartite graph $G=C(K,K',E)$ is an interval graph, then $G_B = B(K,K',E)$ is a difference graph.
\end{lemma}

\begin{proof}
\begin{enumerate}[i)]
\item
Let $u,v$ be the endpoints of $S$. Let $T = \abs{E(S)}+1$ and $r_i$ (resp. $l_j$) be the endpoint of the path $P_i \in \pp_K$ (resp. $P_j \in \pp_{K'}$) among the internal vertices of $S$. Let $a_i = \abs{E(p_S(u,r_i))}$ (resp. $a_j = - \abs{E(p_S(l_j,v))} $) where $p_T(x,y)$ is the unique path between vertices $x$ and $y$ of a tree $T$. By definition, $\abs{a_i} \leq \abs{E(S)} < T$ for every $i \in K \cup K'$. Two paths $P_i \in \pp_K$, $P_j \in \pp_{K'}$ have an edge in common if and only if $\abs{a_i-a_j} \geq \abs{E(S)}+1 = T$. Therefore, $G_B$ is a difference graph.

\item
Fix an interval representation of $G$. For $X \in \set{K,K'}$ let $e_X$ be the edge of the representation that is common to all the paths $\pp_X$ representing the clique $X$.
We can assume without loss of generality that $e_K$ and $e_{K'}$ are the leftmost and rightmost edges of the representation.
We now subdivide $e_K$ and $e_{K'}$ by adding new vertices $v_K$ and $v_{K'}$ respectively.
Finally, if a path contains both $e_K$ and $e_{K'}$ we truncate one edge from its end so that it contains $v_K$ but not $v_{K'}$.
In the new representation, $\pp_K$ and $\pp_{K'}$ meet at the segment between $v_K$ and $v_{K'}$.
\end{enumerate}
\end{proof}

Two representations $\rep$ and $\repprime$ are \emph{bend-equivalent} if they are representations of the same graph $G$ and the paths $P_v \in \pp$ and $P_v' \in \pp'$ representing the same vertex $v$ of $G$ have the same number of bends. We proceed with the following lemma that classifies all the $\boneenpg$ representations of a co-bipartite graph into two types.

\newcommand{\reppr} {\left< H,\pp' \right>}

\begin{lemma}\label{lem:cobipartiteSegments}
Let $G=C(K,K',E)$ be a connected $\boneenpg$ co-bipartite graph with a representation $\rep$. Then $G$ has a bend-equivalent representation $\reppr$ that satisfies exactly one of the following
\begin{enumerate}[i)]
\item $\abs{ \cs(\cup \pp'_K, \cup \pp'_{K'}) }=1$ and $\cup \pp'$ is a tree with maximum degree at most $3$ with at most two vertices of degree $3$ as depicted in Figure~\ref{fig:B1ENPG-cobipartite-representations} (a).
\item $\abs{ \cs(\cup \pp'_K, \cup \pp'_{K'})}=2$ and the paths $\cup \pp'_K$ and $\cup \pp'_{K'}$ intersect as depicted in Figure~\ref{fig:B1ENPG-cobipartite-representations} (b).
\end{enumerate}
\end{lemma}

\begin{proof}
By Proposition~\ref{prop:BendsInClique}, $\cup \pp_K$ and $\cup \pp_{K'}$ are two paths with at most $2$ bends each. Let $e_K$ (resp. $e_{K'}$) be an arbitrary edge of $\cap \pp_K$ (resp. $\cap \pp_{K'}$). The paths $\cup \pp_K$ and $\cup \pp_{K'}$ intersect in at least one edge, because otherwise $G$ is not connected. Therefore, $\abs{\cs(\cup \pp_K, \cup \pp_{K'})} \geq 1$. We consider two disjoint cases:
\begin{itemize}
\item{$\abs{\cs(\cup \pp_K, \cup \pp_{K'})} = 1$.} Let $T=\cup \pp$ and $S$ be the unique segment of $\cs(\cup \pp_K, \cup \pp_{K'})$. The collection $T$ is clearly a tree, since any cycle in $T$ is a cycle in one of $\cup \pp_K, \cup \pp_{K'}$. Any vertex of degree at least $3$ in $T$ is an endpoint of $S$, therefore there are at most $2$ such vertices. On the other hand an endpoint of $S$ has degree at most $3$. Therefore $\Delta(T) \leq 3$ and there are at most $2$ vertices of degree $3$ in $T$.

Let $u$ and $v$ be the two endpoints of $S$. Let also $e_u, e_v$  (respectively $e'_u, e'_v$) be the (at most four) edges not in $S$ but belonging to $\cup \pp_K$ (respectively $\cup \pp'_K$) and incident respectively to $u$ and $v$. Now, shrink all the paths on the branches of $T$ starting with $e_u, e_v, e'_u, e'_v $ and not containing $S$ to respectively the edges $e_u, e_v, e'_u, e'_v $. Clearly, this transformation does not create or remove any split between two paths and does not remove any intersections since all paths intersecting on one of these branches now intersect on the shrunken edge. To maintain bend-equivalence, we add one more edge to every path that loses a bend during the shrinkage (the same edge for all paths of the same branch). This new representation $\reppr$ is bend-equivalent to $\rep$ and $T'=\cup \pp'$ is a tree with the claimed properties.

\item{$\abs{\cs(\cup \pp_K, \cup \pp_{K'})} \geq 2$.} We claim that $\cup \pp_K \cap \cup \pp_{K'}$ contains only horizontal edges, or only vertical edges. Indeed, assume that there is a vertical edge $e_V$ and a horizontal edge $e_H$ in $\cup \pp_K \cap \cup \pp_{K'}$. We observe that there is a unique one bend path with $e_V$ and $e_H$ its end edges, and that any other connecting these edges contains at least three bends. Therefore, both $\cup \pp_K$ and $\cup \pp_{K'}$ contain this path. We conclude that $e_V$ and $e_H$ are in the same segment. As any other edge is either horizontal or vertical, we can proceed similarly for all the edges of $\cup \pp_K \cap \cup \pp_{K'}$ and prove that they all belong to the same segment, contradicting the fact that we have at least $2$ segments. Assume without loss of generality that all the edges of $\cup \pp_K \cap \cup \pp_{K'}$ are vertical. Then every segment is a vertical path. No two segments can be on the same vertical line, because this will require at least one of $\cup \pp_K$, $\cup \pp_{K'}$ to contain four bends. Moreover, three vertical segments in distinct vertical lines imply that $\pp_K$ and $\pp_{K'}$ contain at least four bends each. Therefore, there are exactly $2$ vertical segments and $\pp_K$ (also $\pp_{K'}$) has exactly two bends.
\begin{figure}
\centering
\commentfig{\includegraphics{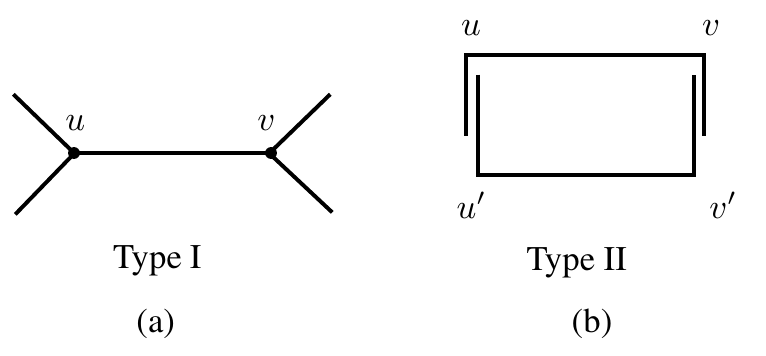}}
\caption{Two types of $\boneenpg$ representation of connected co-bipartite graphs: (a) Type I: $\abs{\cs(K,K')} = 1$, $\cup \pp$ is isomorphic to a tree $T$ with $\Delta(T) \leq 3$ and at most two vertices $u,v$ having degree $3$, (b) Type II: $\abs{\cs(K,K')} = 2$, $\pp_K$ (resp. $\pp_{K'}$) has exactly two bend points $u,v$ (resp. $u', v'$)}
\label{fig:B1ENPG-cobipartite-representations}
\end{figure}

Let $u,v$ (resp. $u',v'$) be the bends of $\cup \pp_K$ (resp. $\cup \pp_{K'}$). Then $\cs(\cup \pp_K, \cup \pp_{K'})$ where $S_u$ (resp. $S_v$) is on the same vertical line as $u$ and $u'$ (resp. $v$ and $v'$). Moreover $e_K$ (resp. $e_{K'}$) is between $u$ and $v$ (resp. $u'$ and $v'$) since otherwise we would have paths crossing both $u$ and $v$ (resp. $u'$ and $v'$) and thus 2 bends. If both the pairs $u,u'$ and $v, v'$ are on different sides of respectively $S_u$ and $S_v$ (as in Figure~\ref{fig:B1ENPG-cobipartite-representations} (b)) then let $H'=H$ and $\pp'=\pp$ be the desired representation. Now consider the situation where $u$ and $u'$ are on the same side of $S_u$.  Every path intersecting $S_u$ crosses the same endpoint of $S_u$ say without loss of generality $u$, implying that if a pair of paths from distinct cliques intersect at $S_u$, they split at this endpoint. Then remove $S_u$ from every path of $\pp_{K'}$ to obtain a bend-equivalent representation that contains one segment. The resulting representation can be transformed into a bend-equivalent representation $\reppr$ as described the previous bullet.
\end{itemize}
\end{proof}

Based on Lemma~\ref{lem:cobipartiteSegments}, a $\boneenpg$ representation of a connected co-bipartite graph $G=C(K, K',E)$ is \emph{Type I} (resp. \emph{Type II}) if $\abs{\cs(K,K')} = 1$ (resp. $\abs{\cs(K,K')} = 2$).

We proceed with the characterization of $\boneenpg$ graphs having a Type II representation that turns out to be simpler than the characterization of the others. In the following lemma, a \emph{trivial} connected component is an isolated vertex.

\begin{lemma}\label{lem:cobipartiteTypeII}
A connected twin-free co-bipartite graph $G=C(K,K',E)$ has a Type II $\boneenpg$ representation if and only if the bipartite graph $G_B=B(K,K',E)$ contains at most two non-trivial connected components each of which is a difference graph.
\end{lemma}

\begin{proof}
$(\Rightarrow)$
Let $\rep$ be a Type II $\boneenpg$ representation of $G$ and $u,v$ (resp. $u', v'$) be the bends of $\cup \pp$ (resp. $\cup \pp'$) as depicted in Figure~\ref{fig:B1ENPG-cobipartite-representations} b). For $x \in \set{u,v}$, let $S_x$ be the segment contained in the path between $x$ and $x'$. The paths of $\pp$ not intersecting with any of $S_u,S_v$ correspond to isolated vertices of $G_B$. Since $G$ is twin-free, there is at most one such path in $\pp_K$ (resp. $\pp_{K'}$).

Each one of the remaining paths intersects exactly one of $S_u$, $S_v$, as otherwise such a path would contain two bends. For $X \in \set{K,K'}$ and $x \in \set{u,v}$ let $\pp_{X_x}$ be the paths of $\pp_X$ intersecting $S_x$. Then $\pp_{K_x}$ and $\pp_{K'_x}$ meet at $S_x$. By Lemma~\ref{lem:meet2difference}, $G_B[K_x \cup K'_x]$ is a difference graph.

$(\Leftarrow)$
It is sufficient to construct a Type II representation for the maximal case, i.e. $G_B$ contains exactly two trivial connected components and two non-trivial connected components. Let $w \in K$ and $w' \in K'$ be the trivial connected components and $B(K_u,K'_u,E_u)$,$B(K_v,K'_v,E_v)$ be the non-trivial connected components of $G_B$.  We construct a rectangle as depicted in Figure~\ref{fig:B1ENPG-cobipartite-representations} b) having vertical lines with $\max(\min(\abs{K_u}, \abs{K'_u}),\min(\abs{K_v}, \abs{K'_v}))+2$ edges, and horizontal lines with one edge $e_K=\set{u,v}$ and $e_{K'}=\set{u',v'}$. For $X \in \set{K,K'}$, and $x \in \set{u,v}$ the paths $\pp_{X_x}$ start with $e_X$ and enter segment $S_x$. The other endpoints of the paths will be in the segment $S_x$. Then, for $x \in \set{u,v}$, $\pp_{K_x}$ and $\pp_{K'_x}$ meet at $S_x$. Since $B(K_x,K'_x,E_x)$ is a difference graph, by Lemma~\ref{lem:difference2meet}, the endpoints can be determined such that $\pp_{K_x} \cup \pp_{K'_x}$ is a representation of $B(K_x,K'_x,E_x)$. The path $P_w$ (resp. $P_{w'}$) consists of the edge $e_K$ (resp. $e_{K'}$). It is easy to verify that this is a representation of $G$.
\end{proof}

We proceed with the characterization of the $\boneenpg$ graphs with a Type I representation. For this purpose we resort to the following definitions from \cite{FHMV04}.

Let $G=B(V,V',E)$ be a bipartite graph and $M \subseteq V \cup V'$. A vertex $v \in V \setminus M$ (resp. $v \in V' \setminus M$) \emph{distinguishes} $M$ if it has a neighbour in $M \cap V'$ (resp. $M \cap V$) and a non-neighbour in $M \cap V'$ (resp. $M \cap V$). A non-empty subset $M$ of $V \cup V'$ is a \emph{bimodule} of $G$ if no vertex distinguishes $M$. It follows from the definition that $V \cup V'$ is a bimodule of $G$, and so are all the singletons and all the pairs of vertices with exactly one from $V$. These bimodules are the \emph{trivial} bimodules of $G$.

A \emph{zed} is a graph isomorphic to a $P_4$ or any induced subgraph of it. We note that a trivial bimodule different from $V \cup V'$ is a zed.

\begin{lemma}\label{lem:cobipartiteTypeI}
A connected twin-free co-bipartite graph $G=C(K,K',E)$ has a Type I $\boneenpg$ representation if and only if there is a set of vertices $Z$ of $G$ such that
\begin{enumerate}[i)]
\item{}\label{item:CoBipartiteTypeI-1} $Z$ is a zed of $G$,
\item{}\label{item:CoBipartiteTypeI-2} $Z$ is a bimodule of $G_B=B(K,K',E)$, and
\item{}\label{item:CoBipartiteTypeI-3} $G_B \setminus Z$ is a difference graph.
\end{enumerate}

Moreover, if $Z$ is a set of vertices of minimum size that satisfies \ref{item:CoBipartiteTypeI-1})-\ref{item:CoBipartiteTypeI-3}) and $Z$ is a set of two non-adjacent vertices of $G$, then for the unique segment $S$ of $\cs(\cup K,\cup K')$ the following hold in every representation $\rep$:
\begin{enumerate}[a)]
\item{}\label{item:CoBipartiteTypeI-A} $S$ is contained in at least one of the paths of $\pp_Z$,
\item{}\label{item:CoBipartiteTypeI-B} the endpoints of $S$ have degree $3$ in $\cup \pp$ and these endpoints constitute $\split(\cup \pp_K,\cup \pp_{K'})$.
\end{enumerate}

\end{lemma}
\begin{proof}
$(\Rightarrow)$
Let $\rep$ be a Type I $\boneenpg$ representation of $G$. By Lemma~\ref{lem:cobipartiteSegments}, $\abs{\cs(K,K')}=1$ and $\cup \pp$ is a tree. Let $u,v$ be the endpoints of the unique segment $S$ of $\cs(K,K')$. We consider the following disjoint cases:

\begin{itemize}
\item {$\set{e_K, e_{K'}} \nsubseteq E(S)$:} Without loss of generality, suppose that $e_K \notin E(S)$ and $u$ is closer to $e_K$ than $v$. Consider two paths $P_{x'},P_{y'} \in \pp_{K'}$ that cross $u$. We observe that these paths are indistinguishable by the paths of $\pp_K$. Namely, every path of $\pp_K$ either does not intersect any one of $P_{x'},P_{y'}$, or intersects both and splits from both at $u$. Therefore the corresponding vertices $x', y'$ are twins. As $G$ is twin-free we conclude that there is at most one path of $\pp_{K'}$ that crosses $u$. Similarly, consider two paths $P_x,P_y \in \pp_K$ that cross $v$. These paths cross also $u$ since $e_K$ is an edge of both paths. Therefore, every path of $\pp_{K'}$ either does not intersect any one of $P_x,P_y$, or intersects both and splits from both at either $u$ of $v$, or intersects both and does not split from any of them. We conclude that there is at most one path of $\pp_K$ that crosses $v$. Let $\pp_{Z'}$ be a set of these at most two paths. Namely, $\pp_{Z'}$ consists of all the paths of $\pp_{K'}$ crossing $u$ and all the paths of $\pp_K$ that cross $v$. We now observe that $\cup (\pp \setminus \pp_{Z'})$ is either a path or the union of two edge-disjoint paths. In both cases no two paths split from each other, and their adjacency is determined only by the intersections. Therefore, the resulting graph $G \setminus Z'$ is an interval graph implying that $G_B \setminus Z'$ is a difference graph. We note that the path $P_{x'} \in \pp_{K'}$ that crosses $u$ is an isolated vertex of $G_B$, therefore for $Z = Z' \setminus \set{x'}$ we have that $G_B \setminus Z$ is a difference graph too, i.e. $Z$ satisfies \ref{item:CoBipartiteTypeI-3}). Since $\abs{Z} \leq 1$, $Z$ satisfies \ref{item:CoBipartiteTypeI-1}) and \ref{item:CoBipartiteTypeI-2}) trivially. This completes the proof of the first part of the claim. As for the second part, since $\abs{Z} \leq 1$, any set of minimum size satisfying the conditions has at most one vertex. Therefore, the second part of the claim holds vacuously.

\begin{figure}
\centering
\commentfig{\includegraphics{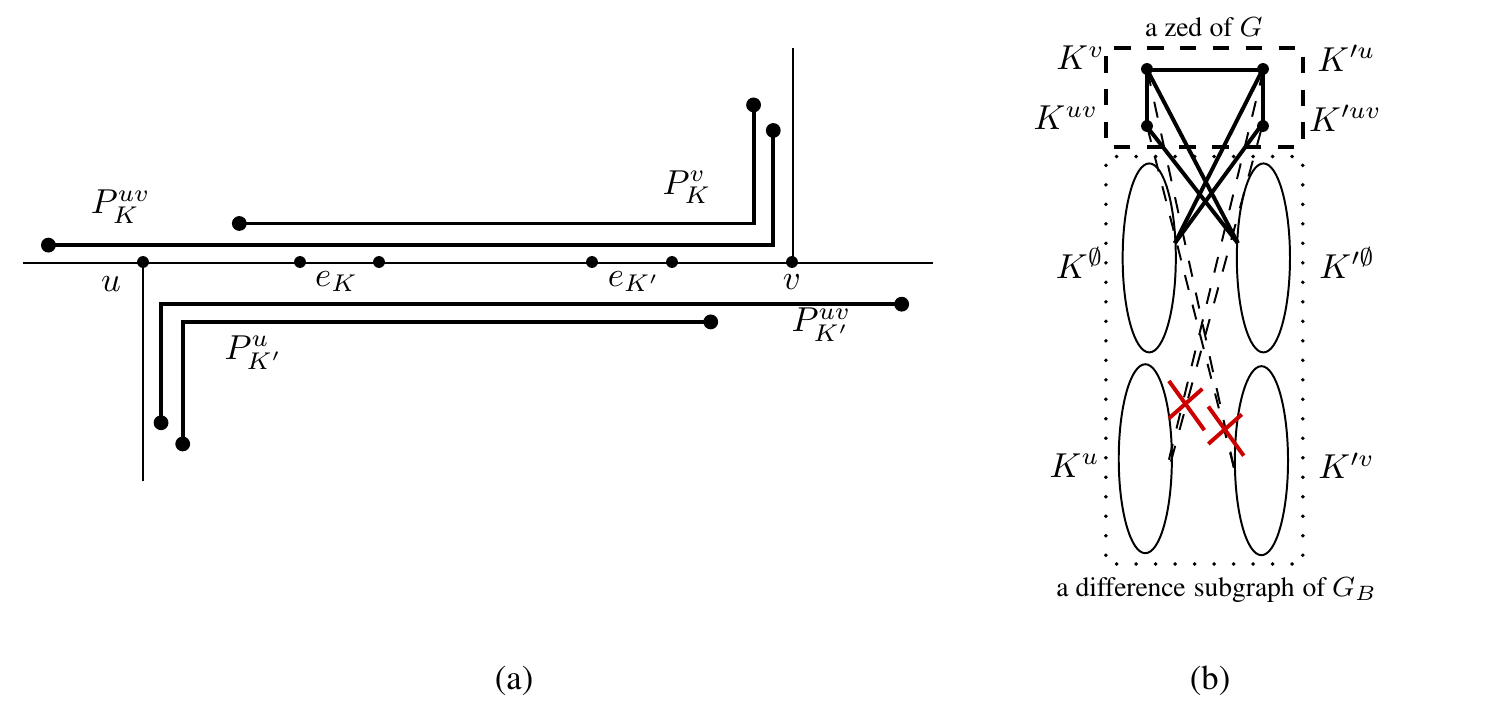}}
\caption{(a) Four special paths corresponding to a zed (b) The type of vertices and edge relations of a $\boneenpg$ co-bipartite graph having a Type I representation. $K^{\emptyset}$ (resp. $K'^{\emptyset}$) is the set of vertices corresponding to the paths of $\pp_K$ (resp. $\pp_{K'}$) crossing neither $u$ nor $v$.}
\label{fig:B1ENPG-cobipartite-TypeI}
\end{figure}

\item {$\set{e_K, e_{K'}} \subseteq E(S)$:} We first note that we can assume $e_K \neq e_{K'}$ since otherwise we can subdivide this edge into two and rename the new edges as $e_K$ and $e_{K'}$. Assume without loss of generality that $e_{K}$ is closer to $u$ than $e_{K'}$, (see Figure~\ref{fig:B1ENPG-cobipartite-TypeI}). Consider two paths $P_{x'},P_{y'} \in \pp_{K'}$ that cross $u$ but not $v$. We observe that these paths are indistinguishable by the paths of $\pp_K$. Therefore, the corresponding vertices are twins. As $G$ is twin-free we conclude that there is at most one path $P_{K'}^u$ of $\pp_{K'}$ that crosses $u$ and does not cross $v$. Similarly there is at most one path $P_K^v$ of $\pp_K$ that crosses $v$ but does not cross $u$, at most one path $P_{K'}^{u,v}$ of $\pp_{K'}$ that crosses both $u$ and $v$, and at most one path $P_K^{u,v}$ of $\pp_K$ that crosses both $u$ and $v$. Let $\pp_Z$ be the set of these at most four paths. As in the previous case, $\cup (\pp \setminus \pp_Z)$  is a path, thus $G_B \setminus Z$ is a difference graph, i.e. $Z$ satisfies \ref{item:CoBipartiteTypeI-3}). Assuming that all the four paths exist, it is easy to verify that their corresponding vertices $K^v,K'^u,K^{u,v},K'^{u,v}$ constitute a $P_4$ with endpoints $K^{u,v},K'^{u,v}$. Therefore, $Z$ is a zed, i.e. $Z$ satisfies \ref{item:CoBipartiteTypeI-1}). Finally, we observe that $P_K^v$ and $P_K^{u,v}$ are distinguishable only by $P_{K'}^u \in \pp_Z$. In other words, they are indistinguishable by paths from $\pp_{K'} \setminus \pp_Z$. By symmetry, we conclude that $Z$ is a bimodule of $G_B$, i.e. it satisfies \ref{item:CoBipartiteTypeI-2}). This concludes the proof of the first part of the claim. To prove the second part, assume by contradiction that there is a minimal set $Z$ satisfying \ref{item:CoBipartiteTypeI-1})- \ref{item:CoBipartiteTypeI-3}) consisting of two vertices and none of the corresponding paths contains the segment $S$. Then these paths are $P^u_{K'}$ and $P^v_K$. We now observe that $P^u_{K'} \sim P^v_K$, i.e. $K^v$ and $K'^u$ are adjacent in $G$, contradicting the assumption that the vertices of $Z$ are non-adjacent in $G$. This concludes the proof of \ref{item:CoBipartiteTypeI-A}). If both paths contain $S$, then these paths are $P^{uv}_K$ and $P^{uv}_{K'}$ and we have $\split(\cup \pp_K,\cup \pp_{K'}) \supseteq \split(P^{uv}_K ,P^{uv}_{K'}) = \set{u,v}$, proving \ref{item:CoBipartiteTypeI-B}) for this case. Otherwise, one of the paths does not contain $S$. Let, without loss of generality this path be $P^u_{K'}$. Then no path of $\pp_{K'}$ crosses $v$. We conclude that $\cup (\pp \setminus \set{P^u_{K'}})$ is a path, implying that the corresponding vertices induce a difference graph on $G_B$, contradicting the assumption that $Z$ is a minimal set satisfying \ref{item:CoBipartiteTypeI-1})-\ref{item:CoBipartiteTypeI-3}).
\end{itemize}
$(\Leftarrow)$
Given a zed $Z$ of $G$ satisfying the conditions of the lemma, we construct a Type I representation $\rep$ as follows. Without loss of generality we assume that $Z$ is a $P_4$ with endpoints $y \in K, y' \in K'$ and internal vertices $x \in K, x' \in K'$. Let $\ell=\min (\abs{K},\abs{K'})+2$. The graph $H$ is a 3 by $\ell+3$ vertices grid where each vertex is represented by an ordered pair from $[-1,\ell+1] \times [-1,1]$. The path $P_x$ (resp. $P_y$) is between $(0,0)$ (resp. $(-1,0)$) and $(\ell,1)$ with a bend at $(\ell,0)$. The path $P_{x'}$ (resp. $P_{y'}$) is between $(\ell,0)$ (resp. $(\ell+1,0)$) and $(0,-1)$ with a bend at $(0,0)$. It is easy to verify that this correctly represents $Z$. The representation of the difference graph $G_B \setminus Z$ is two sets of paths that meet at the line segment between $(0,0)$ and $(\ell,0)$. By Lemma~\ref{lem:difference2meet}, the endpoints of the paths within this segment can be determined in accordance with the difference graph $G_B \setminus Z$. The other endpoints of these paths are determined so as to satisfy the adjacencies of vertices of $Z$ with other vertices, as follows: The other endpoint of every path of $\pp_{K' \cap N_G(y)}$ (resp. $\pp_{K' \setminus N_G(y)}$) is $(\ell,0)$ (resp. $(\ell+1,0)$). The other endpoint of every path of $\pp_{K \cap N_G(y')}$ (resp. $\pp_{K \setminus N_G(y')}$) is $(0,0)$ (resp. $(-1,0)$).
\end{proof}

By Lemmata~\ref{lem:cobipartiteTypeII} and \ref{lem:cobipartiteTypeI} we have the following Theorem.

\begin{theorem}\label{thm:cobipartiteEnpgRecognition}
Let $G=C(K,K',E)$ be a connected, twin-free co-bipartite graph, and $G_B=B(K,K',E)$. Then, $G \in \boneenpg$ if and only if at least one of the following holds:
\begin{enumerate}[i)]
\item $G_B$ contains at most two non-trivial connected components each of which is a difference graph.
\item $G$ contains a zed $Z$ that is a bimodule of $G_B$ such that $G_B \setminus Z$ is a difference graph.
\end{enumerate}
\end{theorem}

Since all the properties mentioned in Theorem~\ref{thm:cobipartiteEnpgRecognition} can be tested in polynomial time we have the following corollary.

\begin{corollary}
$\boneenpg$ co-bipartite graphs can be recognized in polynomial time.
\end{corollary}

\subsection{Efficient Recognition Algorithm}\label{subsec:efficient-recognition}
In this section we describe an efficient algorithm, namely Algorithm~\ref{alg:cobipartiteRecognition}, to recognize whether a co-bipartite graph is $\boneenpg$ using the characterization of Theorem~\ref{thm:cobipartiteEnpgRecognition}. In Algorithm~\ref{alg:cobipartiteRecognition},  \textsc{isTypeI} is a function taking as input a connected twin-free cobipartite graph and a subset $Z$ of vertices to decide if there is $Z'\supseteq Z$ for the graph being $\boneenpg$ of Type I. Similarly, \textsc{isTypeII} takes a connected twin-free cobipartite graph $G$ and returns "YES" if $G$ is  $\boneenpg$ of Type II, and "NO" otherwise. As for function \textsc{FindBimoduleZed}, it takes a twin-free cobipartite graph $G$ and a $Z$ of $G$ to return the minimum superset of $Z$ that is a zed of $G$ and a bimodule of $G_B$, if any. Lastly, the function \textsc{isDifference} in Algorithm~\ref{alg:cobipartiteRecognition} takes a bipartite graph $G$ and either indicates that $G$ is a difference graph or provides a $2K_2$ certifying that $G$ is not a difference graph.

\begin{theorem}\label{thm:cobipartiteEnpgRecognition-timecomplexity}
Given a co-bipartite graph $G=C(K,K',E)$, Algorithm~\ref{alg:cobipartiteRecognition} decides  in time $O(\abs{K}+\abs{K'}+\abs{E})$ whether $G$ is $\boneenpg$.
\end{theorem}

\begin{proof}
Let $n=\abs{K}+\abs{K'}$, $m=\abs{E}$.
Let $T_{diff}(n,m)$ be the running time of $\textsc{isDifference}$ on a graph with $n$ vertices and $m$ edges,
and let $T_{bm}(n,m)$ be the running time of $\textsc{findBimoduleZed}$ that finds a minimum zed of $G$ that is a bimodule of $G_B$ and contains a given zed $Z$. Finally let $\alpha(n,m) \defined T_{diff}(n,m)+T_{bm}(n,m)$.

The correctness of the algorithm follows from Observations~\ref{obs:twin}, \ref{obs:connected}, Lemma~\ref{lem:cobipartiteSegments} and from the correctness of the functions \textsc{isTypeI} and \textsc{isTypeII} that we prove in the sequel.

The correctness of \textsc{isTypeI} is based on Lemma~\ref{lem:cobipartiteTypeI}. A subset $Z$ of vertices of $G$ satisfying \ref{item:CoBipartiteTypeI-1})-\ref{item:CoBipartiteTypeI-3}) of Lemma~\ref{lem:cobipartiteTypeI} is termed as a \emph{certificate} through this proof. We now show that given a twin-free co-bipartite graph $G$ and $Z \subseteq V(G)$, \textsc{isTypeI} returns "YES" if and only if there exists a certificate $Z' \supseteq Z$. Moreover, we show that its running time is at most $5^{5-\abs{Z}} \alpha(n,m)$ when $\abs{Z} \leq 4$ and constant otherwise.

We first observe that if $Z$ is not a zed, then no superset of $Z$ is a zed, and the algorithm returns correctly "NO" in constant time at line \ref{lin:NotAZed}. Therefore, our claim is correct whenever $Z$ is not a zed. We proceed by induction on $5-\abs{Z}$. If $5-\abs{Z}=0$, then $Z$ is not a zed and the algorithm returns "NO" in constant time. In the sequel we assume that $Z$ is a zed. In this case, $Z$ is verified to be a zed by \textsc{isTypeI} in constant time and \textsc{isTypeI} proceeds to line \ref{lin:TestIfIsABimodule} to find (in time $T_{bm}(n,m)$) the minimal bimodule $Z'$ of $G_B$ that contains $Z$ and is a zed of $G$. We consider three cases according to the branching of \textsc{isTypeI}.
\begin{itemize}
\item{$\mathbf{Z' = Z}$ (i.e. $Z$ is a bimodule of $G_B$), \textbf{and} $\mathbf{G_B \setminus Z}$ \textbf{is a difference graph:}} $G_B \setminus Z$ is verified to be a difference graph by \textsc{isTypeI} at line \ref{lin:TestIfIsDifference}. It returns "YES" which is correct by Lemma~\ref{lem:cobipartiteTypeI} since $Z$ is a certificate. The running time is $\alpha(n,m)$, and the result follows since $1 \leq 5^{5-\abs{Z}}$.

\item{$\mathbf{Z' = Z}$ (i.e. $Z$ is a bimodule of $G_B$), \textbf{but} $\mathbf{G_B \setminus Z}$ \textbf{is not a difference graph:}} As $G_B \setminus Z$ is not a difference graph, there is a set $U \subseteq K \cup K' \setminus Z$ such that $G_B[U]$ is a $2 K_2$. Every certificate $Z' \supseteq Z$ must contain at least one vertex of $U$ because otherwise $G_B \setminus Z'$ contains $G_B[U]$ which is a $2K_2$. Therefore, \textsc{isTypeI} proceeds recursively calling \textsc{isTypeI} on $(G, Z \cup \set{u})$ for each $u \in U$. The algorithm returns "YES" if and only if one of the guesses succeeds. Then, the total running time is at most $\alpha(n,m)+4 \cdot 5^{5-(\abs{Z}+1)}\alpha(n,m)<\left(1+4 \cdot 5^{4-\abs{Z}}\right)\alpha(n,m)$. Since $1 \leq 5^{4-\abs{Z}}$ we conclude that the running time is at most $5^{5-\abs{Z}}\alpha(n,m)$.

\item{$\mathbf{Z' \neq Z}$ (i.e. $Z$ is not a bimodule of $G_B$):} If $Z'$ exists, the definition of a bimodule implies that any certificate that contains $Z$ has to contain $Z'$. Therefore, \textsc{isTypeI}$(G,Z')$ is invoked and its result is returned. Otherwise, no certificate contains $Z$ and "NO" is returned. The running time of \textsc{isTypeI} is $T_{bm}(n,m)+5^{5-\abs{Z'}} \alpha(n,m) < (1+5^{5-\abs{Z'}}) \alpha(n,m) \leq 5^{5-\abs{Z}}\alpha(n,m)$.
\end{itemize}

Since \textsc{isTypeI} is invoked initially at line~\ref{lin:InvokeIsType1} with $Z = \emptyset$, together with Lemma~\ref{lem:cobipartiteTypeI} this implies that the algorithm recognizes correctly graphs having a Type I representation. Moreover, the running time of line~\ref{lin:InvokeIsType1} is $5^{5-\abs{\emptyset}} \alpha(n,m)=O(\alpha(n,m))$.

The correctness of \textsc{isTypeII} follows directly from Lemma~\ref{lem:cobipartiteTypeII}. The connected components of $G_B$ can be calculated in $O(n+m)$ time using breadth first search. Therefore, the running time of \textsc{isTypeII} is $O(T_{diff}(n,m))=O(\alpha(n,m))$.

We now calculate the running time of the algorithm. All the twins of a graph can be removed in time $O(n+m)$ using partition refinement, i.e. starting from the trivial partition consisting of one set, and iteratively refining this partition using the closed neighborhoods of the vertices (see \cite{HabibPV99}). Each set of the resulting partition constitutes a set of twins. Summarizing, we get that the running time of Algorithm~\ref{alg:cobipartiteRecognition} is $O(\alpha(n,m))=O(T_{diff}(n,m)+T_{bm}(n,m))$.

$T_{diff}(n,m)$ is $O(n+m)$ (see \cite{HK06}). It remains to prove the correctness of \textsc{findBimoduleZed} and calculate its running time $T_{bm}(n,m)$. We consider the case where $Z$ contains at most one vertex from each one of $K$ and $K'$ and the complementing case where $Z$ contains at least two vertices from $K$ separately.
\begin{itemize}
\item{$Z = \emptyset$ or $Z$ is a singleton or $Z$ is a pair of vertices of $K \times K'$}.
By definition, $Z$ is both a zed of $G$ and a bimodule of $G_B$. Therefore, $Z$ is the minimal bimodule of $G_B$ that is a zed of $G$, and contains $Z$. In this case \textsc{findBimoduleZed} return $Z$ in constant time.

\item{Without loss of generality $Z \cap K$ contains at least two vertices $u_1,u_2$.} We note that $Z \cap K = \set{u_1, u_2}$, because otherwise $Z$ contains a $K_3$ contradicting the fact that it is a zed. Let $Z'$ be the superset of $Z$ obtained by adding to it all the vertices that distinguish $u_1$ and $u_2$. Formally, $Z' \defined (N_{G_B}(u_1) \triangle N_{G_B}(u_2)) \cup Z$. If $Z'$ is not a zed we can return that no superset of $Z$ is both a zed of $G$ and a bimodule of $G_B$. Now, let $Z'$ be a zed and let $U' = Z' \cap K'$. If $\abs{U'} \leq 1$ then $Z'$ is the minimal subset that contains $Z$ and is both a zed of $G$ and a bimodule of $G_B$. If $\abs{U'} > 2$ then $Z'$ is not a zed. Assume $\abs{U'} = 2$ and let $U'=\set{u'_1, u'_2}$. We now add to $Z'$, the set of vertices of $K$ that distinguish $U'$ to get $Z''$. If $Z''=Z'$ then $Z'$ is the minimal superset of $Z$ that is both a zed of $G$ and a bimodule of $G_B$. Otherwise every bimodule that contains $Z'$ has to contain also $Z''$. However $\abs{Z'' \cap K} > \abs{Z \cap K} = 2$, implying that $Z''$ contains a $K_3$, and is thus not a zed. In this case, we conclude that there is no superset of $Z$ as required.
\end{itemize}
As for the running time, we observe that all the operations can be performed in constant time except lines~\ref{lin:FindNeigborhoodDifferenceK} and \ref{lin:FindNeigborhoodDifferenceKPrime} that take time $O(\abs{K'})$ and $O(\abs{K})$, respectively. Therefore, the running time $T_{bm}(n,m)$ of \textsc{FindBimoduleZed} is at most $O(\abs{K}+\abs{K'})=O(n)$. We conclude that the running time of Algorithm~\ref{alg:cobipartiteRecognition} is $O(T_{diff}(n,m)+T_{bm}(n,m))=O(n+m)$.
\end{proof}

We conclude with an interesting remark, pointing to a fundamental difference between $\epg$ and $\enpg$ graphs. A graph is $\bepg{k}$ if it has an $\epg$ representation $\rep$ such that every path of $\pp$ has at most $k$ bends. It is known that given a $\bepg{k}$ representation it is always possible to modify the paths such that every path has exactly $k$ bend; indeed, if there is a path with less than $k$ bends, one can subdivide the edges of the host grid (and consequently all the paths containing the related edges) as needed to introduce new bends until it has exactly $k$ bends, without creating any new intersection or split \cite{GolumbicPC}. The following proposition states that this does not hold for $\benpg{k}$ graphs.

\begin{proposition}\label{prop:3K2containsPathsWithZeroBend}
Every $\boneenpg$ representation of a graph $G = C(K,K',E)$ such that $G_B = B(K,K',E)$ is isomorphic to $3K_2$ contains at least one path with zero bend.
\end{proposition}

\begin{proof}
Consider a set $Z$ consisting of two non-adjacent vertices of $G$. Then $Z$ is a trivial bimodule of $G_B$ and a zed of $G$. Moreover, by Theorem~\ref{theorem:difference-2K2free} $G_B \setminus Z$ is a difference graph since it does not contain a $2K_2$. Therefore, $Z$ satisfies conditions \ref{item:CoBipartiteTypeI-1})-\ref{item:CoBipartiteTypeI-3}) of Lemma~\ref{lem:cobipartiteTypeI}. Then $G$ is $\boneenpg$.

Let $\rep$ be a $\boneenpg$ representation of $G$. Since $G_B$ has three non-trivial connected components, by Lemma~\ref{lem:cobipartiteTypeII}, $\rep$ is a Type I representation. For any single vertex $v$ of $G$, the graph $G_B \setminus \set{v}$ contains a $2K_2$ therefore fails to satisfy condition \ref{item:CoBipartiteTypeI-3}). We conclude that $Z$ is a set of minimum size satisfying the conditions \ref{item:CoBipartiteTypeI-1})-\ref{item:CoBipartiteTypeI-3}) of Lemma~\ref{lem:cobipartiteTypeI}. Moreover, $Z$ consists of two non-adjacent vertices of $G$. Therefore, the unique segment $S$ of $\cs(K,K')$ has the properties~\ref{item:CoBipartiteTypeI-A}) and \ref{item:CoBipartiteTypeI-B}) mentioned in the same Lemma.

\alglanguage{pseudocode}
\begin{algorithm}
\caption{$\boneenpg~\cap$ Co-bipartite Recognition}
\label{alg:cobipartiteRecognition}
\begin{algorithmic}[1]
\Require {A co-bipartite graph $G = C(K,K',E)$ }
\If {$G$ is not connected}
\Return "YES"
\Comment{$G$ has a trivial $\boneenpg$ representation.}
\EndIf

\State Make $G$ twin-free using modular decomposition.

\If {\Call{isTypeI}{$G, \emptyset$}} \Return "YES".\label{lin:InvokeIsType1}
\EndIf
\If {\Call{isTypeII}{$G$}} \Return "YES".\label{lin:InvokeIsType2}
\EndIf
\State \Return "NO".

\Statex
\Function{isTypeI}{$G=C(K,K',E),Z$}
\Require{$G$ is connected, twin-free, $Z \subseteq V(G)$}
\Ensure{returns whether there is a certificate $Z' \supseteq Z$ for $G$ being Type I}
\State $G_B \gets B(K,K',E)$.
\If {$G[Z]$ is not a zed} \Return "NO". \EndIf \label{lin:NotAZed}
\State $Z' \gets$ \Call{findBimoduleZed}{$G,Z$}.\label{lin:TestIfIsABimodule}
\If {$Z' = Z$}
\Comment $Z$ is a zed of $G$ and also a bimodule of $G_B$
    \If {\Call{isDifference}{$G_B \setminus Z$}} \Return "YES". \EndIf \label{lin:TestIfIsDifference}
    \State Let $U \subseteq (K \cup K') \setminus Z$ such that $G_B[U]$ is a $2 K_2$.
    \For {$u \in U$}
        \If {\Call{isTypeI}{$G, Z \cup \set{u}$}} \Return "YES". \EndIf
    \EndFor
    \State \Return "NO".
\Else
    \If {$Z' \neq NULL$}
        \Return \Call{isTypeI}{$G, Z'$}.
    \Else~
        \Return "NO".
    \EndIf
\EndIf

\EndFunction

\Statex
\Function{isTypeII}{$G=C(K,K',E)$}
\Require{$G$ is connected, twin-free}
\Ensure{returns whether $G$ has a Type II representation}
\State $G_B \gets B(K,K',E)$.
\State Remove all isolated vertices from $G_B$.\Comment{There are at most two of them}
\State Calculate the connected components $G_1,\ldots,G_k$ of $G_B$.
\If {$k > 2$} \Return "NO". \EndIf
\If {not \Call{isDifference}{$G_1$}} \Return "NO". \EndIf
\If {not \Call{isDifference}{$G_2$}} \Return "NO". \EndIf
\State \Return "YES".
\EndFunction

\Statex
\Function{findBimoduleZed}{$G=C(K,K',E),Z$}
\Require{$G$ is twin-free, $Z$ is a zed of $G$}
\Ensure{Returns the minimum superset of $Z$ that is a zed of $G$ and a bimodule of $G_B$}
\If {$\abs{Z \cap K} \leq 1$ and $\abs{Z \cap K'} \leq 1$} \Return $Z$. \EndIf
\State Let without loss of generality $Z \cap K = \set{u_1, u_2}$.
\State $Z' \gets (N_{G_B}(u_1) \triangle N_{G_B}(u_2)) \cup Z$.\label{lin:FindNeigborhoodDifferenceK}
\If {$Z'$ is not a zed} \Return NULL. \EndIf
\State $U' \gets Z' \cap K'$.
\If {$\abs{U'} \leq 1$} \Return $Z'$. \EndIf
\State Let without loss of generality $U' = \set{u'_1, u'_2}$.
\State $Z'' \gets (N_{G_B}(u'_1) \triangle N_{G_B}(u'_2)) \cup Z'$.\label{lin:FindNeigborhoodDifferenceKPrime}
\If {$Z''=Z'$}
    \Return $Z'$
\Else~
    \Return NULL.
\EndIf
\EndFunction

\Statex
\Function{isDifference}{$G$} \Comment{\cite{HK06}}
\Require{$G$ is bipartite}
\Ensure{Returns "YES" if $G$ is a difference graph and a $2K_2$ of $G$ otherwise.}
\EndFunction
\end{algorithmic}
\end{algorithm}

Let $Z=\set{x,y'}$ where $x \in K$ and $y' \in K'$, and let $y$ and $x'$ be the unique neighbors in $G_B$ of $x$ and $y'$ respectively. Let also $u,v$ be the endpoints of $S$. By property~\ref{item:CoBipartiteTypeI-A}, without loss of generality $P_x$ contains $S$. Therefore, $P_{x'}$ is contained in $S$ as otherwise it would split from $P_x$ in at least one of $u,v$, contradicting the fact that $x$ and $x'$ are adjacent. By property~\ref{item:CoBipartiteTypeI-B} of the lemma, $u$ and $v$ are split points. To conclude the claim, we now show that $P_{x'}$ has no bends. Assume by contradiction that $P_{x'}$ has a bend $w$. Then $w$ is a bend of $S$ and also of $P_x$. Therefore, $P_x$ does not bend neither at $u$ nor in $v$ as otherwise it would contain $2$ bends. We conclude that both $u$ and $v$ are bends of $\cup \pp_{K'}$. Clearly, $w$ is also a bend of $\cup \pp_{K'}$. Then $\cup \pp_{K'}$ has $3$ bends, contradicting Proposition~\ref{prop:BendsInClique}.
\end{proof}

\section{Summary and Future Work}\label{sec:summary}
In \cite{BESZ14-ENPG-TCS} we showed that $\enpg$ contains an infinite hierarchy of subclasses that are obtained by restricting the number of bends. In this work we showed that $\boneenpg$ graphs are properly included in $\btwoenpg$ graphs. The question whether $\benpg{2} \subsetneq \benpg{3} \subsetneq \ldots$ remains open.

In this work, we studied the intersection of $\boneenpg$ with some special chordal graphs. We showed that the recognition problem of $\boneenpg$ graphs in $\npc$ even for a very restricted sub family of split graphs. On the other hand we showed that this recognition problem is polynomial-time solvable within the family of co-bipartite graphs. A forbidden subgraph characterization of $\boneenpg$ co-bipartite graphs is also work in progress.

We also showed that unlike $\bepg{k}$ graphs that always have a representation in which every path has exactly $k$ bends, some $\boneenpg$ graphs can not be represented using only paths having (exactly) one bend. One can define and study the graphs of edge intersecting non-splitting paths with exactly $k$ bends. Another possible direction is to follow the approach of \cite{CCH16} and consider $\boneenpg$ representations restricted to subsets of the four possible rectilinear paths with one bend.

We showed that trees and cycles are $\boneenpg$. The characterization of their representations is work in progress. A natural extension of such a characterization is to investigate the relationship of $\boneenpg$ graphs and cactus graphs. Another possible extension is to use the characterization of the special case of $C_4$ to characterize induced sub-grids. A non-trivial characterization would imply that not every bipartite graph is $\boneenpg$. Therefore, it would be natural to consider the recognition problem of $\boneenpg$ bipartite graphs. The following interpretation of our results suggests that the latter problem is $\nph$: A clique provides substantial information on the representation, and when the graph is partitioned into two cliques we are able to recognize $\boneenpg$ graphs. However, the absence of one such clique (in case of split graphs) already makes the problem $\nph$. In case of bipartite graphs both of the cliques are absent.



\end{document}